\def\final{0}

\ifnum\final=0
\documentclass[11pt]{article}
\else
\documentclass[prodmode,acmec]{ec-acmsmall} 
\fi

\ifnum\final=0
\usepackage{amsmath, amssymb, amsthm}
\usepackage[affil-it]{authblk}
\usepackage{xspace}
\usepackage{enumitem}
\else
\usepackage{amsmath, amssymb}
\fi

\usepackage{mathtools}
\usepackage{framed}

\ifnum\final=0
\usepackage{fullpage}
\usepackage[margin=1.0in]{geometry}
\fi

\usepackage{verbatim}
\usepackage{color}
\usepackage{caption}
\definecolor{DarkGreen}{rgb}{0.1,0.5,0.1}
\definecolor{DarkRed}{rgb}{0.5,0.1,0.1}
\definecolor{DarkBlue}{rgb}{0.1,0.1,0.5}

\ifnum\final=0
\usepackage{algorithm}
\usepackage{algpseudocode}
\algtext*{EndWhile}
\algtext*{EndFor}
\algtext*{EndIf}
\algnewcommand\algorithmicinput{\textbf{Input:}}
 \algnewcommand\INPUT{\item[\algorithmicinput]}
 \algnewcommand\algorithmicoutput{\textbf{Output:}}
 \algnewcommand\OUTPUT{\item[\algorithmicoutput]}
\else
\usepackage[ruled, noend]{algorithm2e}

\SetAlFnt{\small}
\SetAlCapFnt{\small}
\SetAlCapNameFnt{\small}
\SetAlCapHSkip{0pt}
\IncMargin{-\parindent}

\acmVolume{X}
\acmNumber{X}
\acmArticle{X}
\acmYear{2015}
\acmMonth{2}
\usepackage[numbers, sort&compress]{natbib}
\fi

\newcommand{\Prob}{\mathbb{P}}

\newcommand{\R}{\mathbb{R}}

\newcommand{\G}{\mathcal{G}}
\newcommand{\T}{\mathcal{T}}

\newcommand{\Oh}{\mathcal{O}}
\newcommand{\pttc}{\texttt{PTTC}\xspace}
\newcommand{\rselect}{\texttt{R-SELECT}\xspace}
\newcommand{\cX}{\mathcal{X}}
\newcommand{\cR}{\mathcal{R}}
\newcommand{\cM}{\mathcal{M}}
\newcommand{\bbx}{\mathbf{x}}

\ifnum\final=0
\newtheorem{theorem}{Theorem}[section]
\newtheorem{lemma}[theorem]{Lemma}

\newtheorem{claim}[theorem]{Claim}

\newtheorem{corollary}[theorem]{Corollary}

\theoremstyle{definition}
\newtheorem{definition}[theorem]{Definition}
\fi

\usepackage{natbib}
\usepackage{hyperref}

\ifnum\final=0
\usepackage{graphicx}
\fi

\hypersetup{
    unicode=false,				
    pdftoolbar=true,				
    pdfmenubar=true,			
    pdffitwindow=false, 		
    pdftitle={},    					
    pdfauthor={}
    pdfsubject={},				
    pdfnewwindow=true,		
    pdfkeywords={},			
    colorlinks=true,				
    linkcolor=DarkGreen,		
    citecolor=DarkGreen,	
    filecolor=DarkRed,		
    urlcolor=DarkBlue,		
}

\begin{document}
\ifnum\final=1
\markboth{Sampath Kannan, Jamie Morgenstern, Ryan Rogers, and Aaron
  Roth.}{Private Pareto Optimal Exchange}
\fi

\title{Private Pareto Optimal Exchange}

\ifnum\final=0
\title{Private Pareto Optimal Exchange\thanks{Kannan was partially
    supported by NSF grant NRI-1317788. email: kannan@cis.upenn.edu.
    Morgenstern was partially supported by NSF grants CCF-1116892 and
    CCF-1101215, as well as a Simons Award for Graduate Students in
    Theoretical Computer Science.  Contact information:
    J. Morgenstern, Computer Science Department, Carnegie Mellon
    University, \texttt{jamiemmt@cs.cmu.edu}. Roth was partially
    supported by an NSF CAREER award, NSF Grants CCF-1101389 and
    CNS-1065060, and a Google Focused Research Award. Email:
    \texttt{aaroth@cis.upenn.edu}.}}

\author[1]{Sampath Kannan}
\author[2]{Jamie Morgenstern}
\author[1]{Ryan Rogers}
\author[1]{Aaron Roth}

\affil[1]{Computer Science Department\\
The University of Pennsylvania}
\affil[2]{Computer Science Department\\
Carnegie Mellon University
}
\maketitle
\else
\author{SAMPATH KANNAN
\affil{University of Pennsylvania}
JAMIE MORGENSTERN
\affil{Carnegie Mellon University}
RYAN ROGERS
\affil{University of Pennsylvania}
AARON ROTH
\affil{University of Pennsylvania}}
\fi

\begin{abstract}
  We consider the problem of implementing an individually rational,
  asymptotically Pareto optimal allocation in a barter-exchange
  economy where agents are endowed with goods and preferences over the
  goods of others, but may not use money as a medium of
  exchange. Because one of the most important instantiations of such
  economies is kidney exchange -- where the ``input'' to the problem
  consists of sensitive patient medical records -- we ask to what
  extent such exchanges can be carried out while providing formal
  privacy guarantees to the participants. We show that individually
  rational allocations cannot achieve any non-trivial approximation to
  Pareto optimality if carried out under the constraint of
  differential privacy -- or even the relaxation of
  \emph{joint}-differential privacy, under which it is known that
  asymptotically optimal allocations can be computed in two sided
  markets [Hsu et al. STOC 2014]. We therefore consider a further
  relaxation that we call \emph{marginal}-differential privacy --
  which promises, informally, that the privacy of every agent $i$ is
  protected from every other agent $j \neq i$ so long as $j$ does not
  collude or share allocation information with other agents.  We show
  that under marginal differential privacy, it is possible to compute
  an individually rational and asymptotically Pareto optimal
  allocation in such exchange economies.
\end{abstract}


%
%


\ifnum\final=0
\vfill
\newpage
\tableofcontents
\vfill
\newpage
\else
\maketitle
\fi


\section{Introduction}
Consider the following exchange problem: for each $i\in [n]$, agent
$i$ arrives at a market endowed with a good of type $g_i$ from among a
finite collection of types of goods $\mathcal{G}$, as well as a
\emph{preference over types of goods}, represented by a total ordering
$\succ_i$ over $\mathcal{G}$. In settings where money can be used as
the medium of transaction (and for which people have cardinal
preferences over goods), this would represent an exchange economy for
which we could compute market clearing prices, resulting in a Pareto
optimal allocation. However, in some settings -- most notably markets
for kidney exchanges\footnote{Kidney exchange forms one of the most
  notable ``barter'' markets, but it is not the only example. A number
  of startups such as ``TradeYa'' and ``BarterQuest'' act as market
  makers for barter exchange of consumer goods.} \cite{RSU05} -- the
use of money is not permitted, and agents are limited to participating
in an exchange -- essentially a permutation of the goods amongst the
agents, with no additional payments. In such settings, we may still
require that:

\begin{enumerate}
\item The exchange be \emph{individually rational} -- i.e. every agent
  weakly prefers the good that they receive to the good that they were
  endowed with, and
\item the exchange be (approximately) \emph{Pareto optimal} -- that
  there should not be any other permutation of the goods that some
  agent (or, in the approximate case, many agents) strictly prefers,
  unless there exists another agent who strictly prefers the original
  allocation.
\end{enumerate}
Because one of the key applications of efficient barter-exchange
involves computation on extremely sensitive medical data, this paper
investigates to what extent it can be accomplished while guaranteeing
a formal notion of \emph{privacy} to the agents participating in the
market, without compromising the above two desiderata of individual
rationality and (approximate) Pareto optimality. We wish to give
algorithms that protect the privacy of agents' initial endowment as
well as their preferences over goods from the other agents in the
market, using the tools of differential privacy.

In this respect, our paper continues a recent line of work exploring
the power of algorithms satisfying various privacy definitions
(relaxations of \emph{differential privacy}) in different kinds of
exchange problems. To understand this question, note that the
\emph{input} to an algorithm clearing an exchange economy is
partitioned amongst the $n$ agents (each agent reports her initial
endowment $g_i$ and her preference ordering $\succ_i$), as is the
output (the mechanism reports to each agent the type of good she will
receive, $g_i'$.) This allows us to parameterize privacy guarantees in
terms of adversaries who can see differing sets of outputs of the
mechanism. The standard notion of \emph{differential privacy},
requires that we guarantee privacy even against an adversary who can
see all $n$ outputs of the mechanism -- i.e. the type of good that is
allocated to each of the $n$ agents. It is intuitively clear that
nothing non-trivial can be done while guaranteeing differential
privacy in allocation problems -- ``good'' allocations must give
individuals what they want, and this is exactly what we must keep
private. An adversary who is able to see the allocation given to
agent $i$ by any mechanism which guarantees individual rationality
would immediately learn the relative ranking of the good that agent
$i$ was allocated compared to the good that he was endowed
with. However, this does not rule out the possibility of protecting
the privacy of agent $i$ against an adversary who can only see the
allocation of \emph{some} agents-- notably, not the allocation given
to agent $i$.

\subsection{Our Results}

The question of privately computing allocations where agents'
preferences are sensitive data was first studied by~\citet{Matching},
who showed that in two sided allocation problems (with distinguished
\emph{buyers} and \emph{sellers}) with monetary transfers, no
non-trivial allocation can be computed under the constraint of
differential privacy, when we must protect the privacy of the
buyers. However,~\citet{Matching} showed that near-optimal results can
be achieved under \emph{joint differential privacy}, which informally
requires that for every agent $i$ simultaneously, the joint
distribution on allocations given to agents $j \neq i$ be
differentially private in the data of agent $i$. This corresponds to
privacy against an adversary who can see the allocation of
\emph{all other agents} $j \neq i$, but who cannot observe agent
$i$'s own allocation when trying to violate $i$'s privacy.

The allocation problem we study in this paper is distinct from the two
sided problem studied in~\cite{Matching} in that there are no
distinguished buyers and sellers -- in our barter exchange problem,
every agent both provides a good and receives one, and so we must
protect the privacy of every agent in the market. We insist on
algorithms that always guarantee \emph{individually rational}
allocations, and ask how well they can approximate \emph{Pareto
  optimality}. (Informally, an allocation $\pi$ is
$\alpha$-approximately Pareto optimal if for every other allocation
$\pi'$ that is strictly preferred by an $\alpha$-fraction of agents,
there must be some other agent who strictly prefers $\pi$ to $\pi'$.) We
start by showing that even under the relaxed notion of \emph{joint}
differential privacy, no individually rational mechanism can achieve
any nontrivial approximation to Pareto optimality (and, since joint
differential privacy is a relaxation of differential privacy, neither
can any differentially private mechanism).
\begin{theorem}[](Informal)
{\it No $\epsilon$-jointly differentially private algorithm for the
exchange problem that guarantees individually rational allocations
can guarantee with high probability that the resulting allocation will
be $\alpha$-approximately Pareto optimal for:}
\[\alpha \leq 1 - \frac{e^\epsilon}{e^\epsilon +1}\]
\end{theorem}

Given this impossibility result, we consider a further relaxation of
differential privacy, which we call \emph{marginal differential
  privacy}. In contrast to joint differential privacy, marginal
differential privacy requires that, simultaneously for every pair of
agents $i \neq j$, the \emph{marginal} distribution on agent $j$'s
allocation be differentially private in agent $i$'s data. This
corresponds to privacy from an adversary who has the ability only to
look at a single other agent's allocation (equivalently -- privacy
from the other agents, assuming they do not collude). Our main result
is a marginally-differentially private algorithm that simultaneously
guarantees individually rational and approximately Pareto optimal
allocations, showing a separation between marginal and joint
differential privacy for the exchange problem:
\begin{theorem}[](Informal)
 There exists an $\epsilon$-marginally differentially private
  algorithm that solves the exchange problem with $n$ agents and $k =
  |\mathcal{G}|$ types of goods by producing an allocation which is
  individually rational and, with high probability,
  $\alpha$-approximately Pareto optimal for
\[\alpha = O\left(\frac{\mathrm{poly}(k)}{\epsilon n}\right)\]
\end{theorem}

Note that the approximation to Pareto optimality depends polynomially
on the number of \emph{types} of goods in the market, but decreases
linearly in the number of participants in the market $n$. Hence,
fixing $k$, and letting the market size $n$ grow, this mechanism is
asymptotically Pareto optimal.

It is natural to ask whether this bound can be improved so that the
dependence on $k$ is only $\alpha =
O\left(\frac{\mathrm{polylog}(k)}{\epsilon n}\right)$, which is the
dependence on the number of distinct types of goods achieved in the
approximation to optimality in \cite{Matching} (again, under
\emph{joint} differential privacy, for a two-sided market). We show
that this is not the case.
\begin{theorem}[](Informal)
 For every $\epsilon$-marginally differentially
  private algorithm that on every instance of the exchange problem
  with $n$ agents and $k = |\mathcal{G}|$ types of goods, produces an
  individually rational allocation that with high probability is
  $\alpha$-approximately Pareto optimal, we have:
  $$\alpha = \Omega\left(\frac{k}{n}\left(1- \frac{e^\epsilon}{e^\epsilon+1} \right)\right)$$
\end{theorem}

We also consider
exchange markets in which every agent brings exactly one good to the market, but
there are also a small number of \emph{extra} goods available to distribute
that are not brought to the market by any of the agents.  In the kidney exchange
setting, these can represent altruistic donors or non-living donors who
provide kidneys for transplantation, but do not need to receive a kidney in
return. The existence of extra goods intuitively makes the problem easier,
and we show that this is indeed the case. In this setting we use techniques
from \cite{Matching, HHRW14} to get an individually rational, asymptotically
exactly Pareto optimal allocation subject to \emph{joint differential
privacy} (i.e. circumventing our impossibility result above), under the
condition that the number of extra copies of each type of good is a (slowly)
growing function of the number of agents.


\subsection{Related Work}
Differential privacy, introduced by~\citet{DMNS06} has become a
standard ``privacy solution concept'' over the last decade, and has
spawned a vast literature too large to summarize. We here discuss only
the most closely related work.

Although the majority of the differential privacy literature has
considered numeric valued and continuous optimization problems, a
small early line of work including~\citet{NRS07}
and~\citet{GLMRT10} study combinatorial
optimization problems. The dearth of work in this area in large part
stems from the fact that many optimization problems cannot be
nontrivially solved under the constraint of differential privacy,
which requires that the entire output be insensitive to any
input. This problem was first observed by~\citet{GLMRT10} in the
context of set cover and vertex cover, who also noted that if the
notion of a solution is slightly relaxed to include private
``instructions'' which can be given to the agents, allowing them
(together with their own private data) to reconstruct a solution, then
more is possible. Similar ideas are also present in~\citet{MM09}, in
the context of recommendation systems.

\emph{Joint} differential privacy, which can be viewed as a
generalization of the ``instructions'' based solution
of~\citet{GLMRT10}, was formalized by~\citet{Large}, who showed that,
although correlated equilibria in large games could not be computed to
any nontrivial accuracy under differential privacy, they can be
computed quite accurately under joint differential privacy. A similar
result was shown by~\citet{RR14} for \emph{Nash} equilibria in
congestion games.~\citet{Matching} subsequently studied a two-sided
allocation problem, in which buyers with private valuation functions
over bundles of goods must be allocated goods to maximize social
welfare (note that here buyers are allocated goods, but do not provide
them, unlike the problem we study in this work). They also show that,
although the allocation problem cannot be solved to nontrivial
accuracy under differential privacy, it can be solved accurately (when
buyers preferences satisfy the \emph{gross substitutes} condition) under
joint differential privacy.

In the present paper, we continue the study of private allocation
problems, and consider the exchange problem in which $n$ agents both
supply and receive the goods to be allocated. This problem was first
studied by~\citet{TTC}, who also proposed the Top Trading Cycles
algorithm for solving it (attributing this algorithm to David
Gale). We show that this problem is strictly harder from a privacy
perspective than the two sided allocation problem: it cannot be solved
non-trivially even under joint differential privacy, but can be solved
under a weaker notion which we introduce, that of marginal
differential privacy. Our solution involves a privacy-preserving
modification of the top trading cycles algorithm. To the best of our
knowledge, we are the first to give \emph{marginal} differential
privacy a name and to demonstrate a separation from joint differential
privacy. The solution concept has, however, appeared in other works --
for example in~\citet{HM14}, in the context of privacy preserving and
incentive compatible recommendation systems.

\section{Model}\label{sec:model}
We study the setting of trading within an exchange market where there are
$k$ types of goods. We denote this set of types of goods
as $\G$. The set of agents will be denoted as $N$ where $|N| = n$. Each $i\in N$
has one copy of some type of good $g_i\in \G$, and some strict linear
preference $\succ_i$ over all good types in $\G$.  Let $N_j =  \{i : g_i =
j \}  $ and $n_j = |N_j|$ denote the set and number of agents, respectively, who bring good $j\in \G$ to the
market. Since each agent brings exactly one good to the market, we have
$\sum_{j\in \G} n_j = n$. An instance of an exchange market is
given as $\mathbf{x} = (x_i)_{i \in N}$ where each $x_i =
(g_i,\succ_i)$.  Our goal will be to find beneficial trades amongst
the agents in the market.

\begin{definition}[Allocation]
  An \emph{allocation} is a mapping $\pi:N \to \G$ where we have $|\{i\in N: \pi(i) = j
  \}| = n_j$ for each $j \in \G$.
\end{definition}

We say an algorithm is \emph{individually rational} if each agent is allocated a type of good she weakly prefers to her initial endowment. Formally:

\begin{definition}[IR]
  An allocation $\pi$ is Individually Rational (IR) if $\pi(i)
  \succeq_i g_i$ $\forall$ $i \in N$.
\end{definition}

IR alone does not ensure that high-quality solutions are found: in particular, one IR allocation rule is to make no trades, i.e. $\pi(i) = g_i$ for all $i \in [n]$. It does, however, ensure that no agent is worse off for having participated in the mechanism. Pareto optimality (PO), on the other hand, gives a way to describe the inherent quality of an allocation. An allocation is PO if it
cannot be changed to improve some agent $i$'s utility without harming
the utility of some other agent $j$. Under privacy constraints, it
will be impossible to obtain exact PO, and so we
 instead ask for an approximate version.

\begin{definition}[$\alpha$-PO]
  An allocation $\pi$ is $\alpha$-\emph{approximately Pareto optimal} (or just $\alpha$-PO) if for
  any other allocation $\pi'$, if there exists a set $S \subset N$
  with $|S|>\alpha n$ such that $\pi'(i) \succ_i \pi(i), \forall i \in
  S$, then there must be some $j \in N\backslash S$ such that
  $\pi(j) \succ_j \pi'(j)$. In other words, an allocation is
  $\alpha$-PO if strictly improving the
  allocation for more than an $\alpha$-fraction of agents necessarily requires
  strictly harming the allocation of at least $1$ agent.

  We say that an algorithm guarantees $\alpha$-PO if, on every exchange problem instance, it outputs an
  $\alpha$-PO allocation. If $\alpha =
  \alpha(n)$ is a function of the number of agents $n$, we say that an
  algorithm is asymptotically Pareto optimal if it guarantees
  $\alpha(n)$-PO, and $\alpha(n) = o(1)$.
\end{definition}

We wish to compute such allocations while guaranteeing a formal notion
of privacy to each of the participating agents. The notions of privacy we
consider will all be relaxations of \emph{differential privacy}, which
has become a standard privacy ``solution concept''. We borrow standard
notation from game theory: given a vector $\bbx \in \cX^n$, we write
$\bbx_{-i} \in \cX^{n-1}$ to denote the vector $\bbx$ with the $i$th
coordinate removed, and given $x'_i \in \cX$ we write $(\bbx_{-i},x'_i)
\in \cX^n$ to denote the vector $\bbx$ with its $i$th coordinate
replaced by $x'_i$.
\begin{definition}[Differential Privacy,~\citep{DMNS06}]
  A mechanism $M:\cX^n \to R$ satisfies $(\epsilon,\delta)$-
  differential privacy if for every $i \in [n]$, for any two types $x_i,x_i' \in \cX$, any
  tuple of types $\bbx_{-i} \in \cX^{n-1}$, and any $B \subseteq R$, we
  have
$$
\Prob\left( M(\bbx_{-i},x_i)\in B \right) \leq e^\epsilon \Prob\left( M(\bbx_{-i},x'_i) \in B  \right) + \delta
$$
\end{definition}

Here $R$ denotes an arbitrary range of the mechanism. The
definition of differential privacy assumes that the \emph{input} to
the mechanism $\bbx$ is explicitly partitioned amongst $n$ agents. In
the problem we consider, the sensitive data is the exchange market $\mathbf{x}$, so we simultaneously want to preserve the privacy of each agent, in the form of his initial endowment as well as his preferences over other goods.  Note that the input $\mathbf{x}$ is partitioned and the
range of the mechanism is also naturally partitioned between $n$
agents (the output of the mechanism can be viewed as $n$ messages, one
to each agent $i$, telling them the type of good they are receiving,
$\pi(i)$).  In such cases, we can consider relaxations of differential
privacy informally parameterized by the maximum size of collusion that
we are concerned about. Joint differential privacy, defined
by~\citet{Large}, asks that the mechanism simultaneously protect the
privacy of every agent $i$ from arbitrary collusions of up to $n-1$
agents $j \neq i$ (who can share their own allocations, but cannot see
the allocation of agent $i$):

\begin{definition}[Joint Differential Privacy,~\citep{Large}]
  A mechanism $M:\cX^n \to O^n$ satisfies $(\epsilon,\delta)$-joint
  differential privacy if for any agent $i \in [n]$, any two types
  $x_i,x_i' \in \cX$, any tuple of types $\bbx_{-i} \in \cX^{n-1}$, and
  any $B_{-i} \subseteq O^{n-1}$, we have
$$
\Prob\left( M(\bbx_{-i},x_i)_{-i} \in B_{-i} \right) \leq e^\epsilon \Prob\left( M(\bbx_{-i},x'_i))_{-i} \in B_{-i}  \right) + \delta
$$
\end{definition}

As we will show, it is not possible to find IR and
asymptotically Pareto optimal allocations under joint differential
privacy, and so we introduce a further relaxation which we call
marginal differential privacy. Informally, marginal differential
privacy requires that the mechanism simultaneously protect the privacy
of every agent $i$ from every other agent $j \neq i$, assuming that
they do not collude (i.e. it requires the differential privacy
condition only on the \emph{marginal} distribution of allocations to
other agents, not on the joint distribution).

\begin{definition}[Marginal Differential Privacy]
\label{def:mdp}
A mechanism $M: \cX^n \to O^n$ satisfies $(\epsilon,\delta)$-marginal differential privacy if $\forall i \neq j$, $\forall \bbx_{-i} \in \cX^{n-1}$, $\forall x_i,x_i' \in \cX$, and $\forall B \subset O$, we have
$$
\Prob(M(\bbx_{-i},x_i)_j \in B) \leq e^\epsilon \Prob(M(\bbx_{-i}, x_i')_j \in B) + \delta
$$
\end{definition}

\section{Lower Bounds}\label{SEC:LB}
In this section, we show lower bounds on how well the exchange market
problem can be solved subject to privacy constraints. We first show
that under the constraint of joint-differential privacy, there does
not exist any IR and asymptotically Pareto optimal
mechanism. This motivates our relaxation of studying exchange problems
subject to marginal differential privacy. We then show that under
marginal differential privacy, any mechanism producing IR and
$\alpha$-PO allocations must have $\alpha = \Omega(k/n)$, where $k$ is
the number of distinct types of goods (that is, a linear dependence on
$k$ is necessary). We complement these impossibility results in
Section~\ref{SEC:PTTC}, where we show that under marginal differential
privacy, it is indeed possible to achieve both IR and $\alpha$-PO
simultaneously, if the number of good types $k$ satisfies $k = o(n^{2/9})$.

Our impossibility result for joint differential privacy is based on a
reduction to the following well-known claim, which we prove in
Appendix \ref{section:lb_proofs}.

\begin{claim}\label{claim:bit}
  There is no $(\epsilon,\delta)$- differentially private mechanism $M: \{
  0,1\} \to \{ 0,1\}$ such that $\Prob(M(b) = b)>
  \frac{e^\epsilon+\delta}{e^\epsilon + 1}$ for both $b = 0, 1$.
\end{claim}

In the setting of exchange markets, we let
$\mathcal{X} = \G \times \T$ where $\G$ is the set of good types and
$\T$ is the set of linear orderings over $\G$ for a single
agent. Now, we show that it will not be possible to guarantee
privacy, IR, and $o(1)$-PO (with constant $\epsilon$).

\begin{theorem}\label{thm:lowerbound}
  For any $\epsilon,\delta,\beta>0$, if we have an $(\epsilon,\delta)$-joint differentially private
  mechanism $M_J: \mathcal{X}^n \to \G^n$ that guarantees an
  $\alpha$-PO allocation with probability at
  least $1-\beta$ and always gives an IR allocation then
\[\alpha \geq 1- \frac{e^\epsilon+\delta}{(1-\beta)(e^\epsilon+1)} \]
\end{theorem}
\proof[ \ifnum\final=0 Proof \fi Sketch] The formal proof is provided in
Appendix~\ref{section:lb_proofs}; here we sketch the basic intuition.
We consider an exchange market with two types of goods, 0 and 1, and
$n$ agents endowed with each type of good. An adversary who has
control over $n$ agents with good $1$ and $n-1$ agents with good $0$
can determine, by observing the trades of the agents she controls,
whether the remaining agent with good $0$ traded. When she does, the
adversary can determine (reasoning based on the fact that the
mechanism is IR) that she strictly preferred good $1$ to good $0$. We
can then use Claim~\ref{claim:bit} to upper-bound the probability that
she is allowed to trade in this event. The mechanism does not know to
distinguish between the identified agents and the $2n-1$ agents
controlled by the adversary, so reasoning about the case in which this
identified agent is selected uniformly at random gives the desired
bound.
\endproof

We also show in Appendix \ref{section:lb_proofs} that even under
marginal differential privacy, the approximation to Pareto optimality
must have at least a linear dependence on $k$.

\begin{theorem}\label{thm:lb-marginal}
  Any $(\epsilon,\delta)$-marginally differentially private
  mechanism $M_S: \mathcal{X}^n \to \G^n$ that guarantees an
  $\alpha$-PO allocation with probability at
  least $1-\beta$ and always satisfies IR must have:
\[\alpha \geq \frac{k(1-\beta)}{n}\left(1-\frac{e^\epsilon+\delta}{e^\epsilon+1}\right)\]
\end{theorem}
\proof[\ifnum\final=0 Proof \fi Sketch] Consider an instance with $k$ types of goods and we
write $\G = [k]$. Each agent $i\in\{1, \ldots, k-1\}$ is endowed with
good $g_i=i$ and her most preferred good is $g_i+1$, followed by
$g_i$.  Then we let agents $k$ through $n$ have good type $k$, all of
which prefer their endowment the most except agent $k$ who may prefer
a different type of good.  Any IR allocation must
give agent $i \in \{1, \ldots, k-1\}$ either good $g_i+1$ or $g_i$ and
agents $\{k+1,\cdots, n\}$ good type $k$.
  
Agent $k$, endowed with good $g_k=k$, is the subject of attack. If
she prefers good $g_1$, then $k$ goods can be cleared, giving each
person her favorite good (agent $i$ getting good $g_i+1$, for each
$i\in [1, \ldots, k-1]$ and $k$ getting $g_1$). If not, no trades can
be made. Thus, any agent $i\in[1, \ldots, k-1]$ knows with certainty
that $k$ prefers $g_1$ if she receives good $g_i+1$. With this
construction, we again reduce to Claim~\ref{claim:bit} to lower
bound the probability that the mechanism is allowed to clear the
$k$-cycle, even when it is possible.
\endproof

The proof of this lower bound also implies that asymptotic Pareto
optimality requires there to be many copies of at least one type of
good, with a proof given in Appendix \ref{section:lb_proofs}.

\begin{corollary}
  Any $(\epsilon,\delta)$-marginally differentially private mechanism
  that gives an IR allocation with certainty and an $\alpha$-PO
  allocation with probability $1-\beta$ requires there to be at least
  one type of good with $c$ copies where
$$
c\geq\frac{(1-\delta)(1-\beta)}{\alpha(e^\epsilon+1)}
$$
\label{cor:copies}
\end{corollary}

\section{Private Top Trading Cycles}\label{SEC:PTTC}
In this section we describe an algorithm which computes IR and
$\alpha$-PO exchanges, subject to marginal differential privacy.  The
algorithm $\pttc$ (given formally in Algorithm \ref{PTTC} in Appendix
\ref{sec:appendix}) takes an exchange market
$\mathbf{x} = (g_i, \succ_i)_{i=1}^n$ as input and allocates a good to
each agent.  $\pttc$ begins by considering the complete directed graph
$G = (V,A)$ with each node $u \in V$ corresponding to a type of good
in $\G$.  For each arc $e= (u,v) \in A$ we define $P_{(u,v)}$ to be
the set of people who have good $u$ and consider $v$ their favorite
remaining good (and thus want to trade along arc $e=(u,v))$, i.e.
$P_{(u,v)} = \{i \in N : g_i = u \quad \& \quad v \succeq_i g \quad
\forall g \in V \}.$  We then define the arc weight $w_e$ for arc $e \in A$ as
$w_e = \left| P_e \right|.$ 

Let $N_u$ be the set of people that are
endowed with good $u$, so $N = \cup_{u \in V} N_u$. The node weights
$n_u$ for $u \in V$ are then $n_u = |N_u|.$
We first give a high level presentation of our algorithm $\pttc$ that
works as follows:
\begin{enumerate}
\item We compute the arc-weights, and then add noise to them (to preserve
  privacy).
\item We ``clear'' cycles that have sufficient weight on each
  arc. Note this determines which goods will trade, but not yet
  \emph{who} specifically will trade. Because of the added noise, this
  step satisfies differential privacy.
\item We actually execute the trade. This involves selecting exactly
  \emph{who} will trade along the cycle selected above. We select the
  individuals at random from among the set that want to trade.  
  Because we need the number of agents trading on each
  arc in the cycle to be identical, this step
  is only marginally differentially private.
\item We remove any goods that have few remaining copies (and allocate
  any remaining agent possessing that good his own good), and then
  return to step (1).
\end{enumerate}

Now, we describe the algorithm in more formality, along with defining
the necessary parameters.  Note that we will use an ordering of the
$n$ agents in our algorithm, say by their index.  $\pttc$ works as
follows:

\begin{enumerate}
\item (Arc Weight Noise) With failure probability $\beta>0$ and
    privacy parameters $\epsilon,\delta_1,\delta_2>0$, we add Laplace noise $Z_e$
    with parameter $1/\epsilon'$ to each arc weight $w_e$ where
\begin{equation}
\epsilon' =\frac{\epsilon \log(k^3/\beta)}{2\sqrt{8}\left( \log(k^3/\beta)\sqrt{k\log(1/\delta_1)} + k\sqrt{k\log(1/\delta_2)} \right) }.
\label{eq:eps_prime}
\end{equation}
  We let $E$ denote a high probability upper
  bound (see Claim \ref{lowerror}) on all realizations of the noise $Z_e$, which is
\begin{equation}
E = \frac{\log(k^3/\beta)}{\epsilon '}.
  \label{eq:error}
\end{equation}
 We then set the noisy arc weights to be $\hat w_e = \max\{w_e + Z_e - 2E, 0\}$.
\item (Clear Cycle) If there exists a cycle $C$ in the graph where all
  the arcs along it have weight $\lfloor \hat w_e \rfloor > 0$, then
  set $\hat W = \min_{e \in C} \{\lfloor \hat w_e \rfloor \}$ as the
  weight of $C$.  Otherwise there is no cycle and we skip to step (4).
\item (Trade) With the cycle $C$ of weight $\hat W$ given in step (2), for
    each $e \in C$, we select $\hat W$ agents at random from $P_e$ to trade
    along $e$ and denote this set of \emph{satisfied} agents $S_e$.  The
    way we choose the satisfied agents is via the $\rselect(\hat W,P_e)$ procedure
    given in Algorithm \ref{R-SELECT} which ensures that any fixed agent
    at arc $e \in C$ is selected with probability $\hat W/w_e$.  We then update, for every $e =
    (u,v) \in C$:
  \[P_e \gets P_e \backslash S_e \quad w_e \gets w_e - \hat W \quad
  \hat w_e \gets \hat w_e - \hat W\]
  (i.e. remove the agents served and update weights) and for
  each $i\in S_e$, set $\pi(i) = v$ where $g_i = u \text{ and } e =
  (u,v)$ (they get their requested good).  Return to step (2).
\item (No Cycle - Deletion) If there is no cycle with positive weight, then
    there must be a node $v$ such that $ \hat n_v = \sum_{e:e=(v,z),z\in V}\hat
    w_e < k $.  We then update $V \gets V\backslash\{v\}$.  The arcs that
    were incoming $IN_v = \{e \in A : e = (u,v) \text{ for some } u \in V
    \}$ and outgoing from $v$, $OUT_v = \{e \in A: e = (v,u) \text{ for some
    } u \in V \}$ are also deleted: $A \gets A \backslash \{OUT_v \cup
    IN_v\}.$

    We then move all agents in $P_e$ for each $e \in IN_v$ to the arc
    that corresponds to trades between $g_i$ and their next most
    preferred good among the remaining vertices.  For every
    $e = (u,w) \in A$ we define these added arcs as
\[B_{(u,w)} = \{ i \in P_{(u,v)} :  w \succeq_i g \quad \forall g \in V \}\]
 (where $v\notin V$) and update the quantities
$w_e \gets w_e + |B_e|$ and $P_e \gets P_e \cup B_e.$

We then assign $v$ to all agents whose endowment was $v$ and have not
been assigned a good (for all $i$ such that
$i \in P_e \text{ where } e \in OUT_v,$ set $\pi(i) =v).$ If any nodes
remain, return to Step (1).

\item (Clean-up for IR) If we ever assign too many goods or select
  more agents to trade than actually want to trade, i.e. some
  $\hat w_e$ does not satisfy \eqref{eq:webound}, then undo every
  assignment of the goods and give everyone the good they started
  with.
\end{enumerate}


\ifnum\final=0
\begin{algorithm}
\caption{Randomly Select Agents}\label{R-SELECT}
\begin{algorithmic}[0]
\INPUT : Number to select $W$ and ordered set of agents $P$.
\OUTPUT : $W$ randomly selected agents from $P$.
\Procedure {$\rselect$} {$W,P$}
\State Set $M \gets |P|$ 
\State Relabel individuals in order of $P$ as $1,2,\cdots, |P|$
\State For some large integer $Z>>n\qquad \backslash \backslash$ For example, set $Z$
  to be the population of Earth. 
\State Select $\sigma \in [Z]$ uniformly at random
\State Set $S\gets \{\sigma+1,\sigma+2,\cdots,\sigma+W \mod |P|\}$ 
\State Map back all labels in $S$ to the original labels they had in $P$. \\
\Return $S$.

\EndProcedure
\end{algorithmic}
\end{algorithm}

\else

\begin{algorithm}[h]
\SetAlgoNoLine
$\rselect$ ($W,P$) \;
\KwIn{Number to select $W$ and ordered set of agents $P$.}
\KwOut{$W$ randomly selected agents from $P$ }
{
Set $M \gets |P|$ \newline
Relabel individuals in order of $P$ as $1,2,\cdots, |P|$\newline
For some large integer $Z>>n\qquad \backslash \backslash$ For example, set $Z$
  to be the population of Earth. \newline
Select $\sigma \in [Z]$ uniformly at random\newline
Set $S\gets \{\sigma+1,\sigma+2,\cdots,\sigma+W \mod |P|\}$ \newline
Map back all labels in $S$ to the original labels they had in $P$. \newline
\Return $S$.
}
\caption{Randomly Select Agents}\label{R-SELECT}
\end{algorithm}
\fi


We state the formal version of our main theorem below, which the rest of the paper is dedicated to proving.
\begin{theorem} \label{thm:mdp} For
  $\epsilon,\delta_1,\delta_2,\beta>0$, our algorithm \pttc is
  $(\epsilon, \delta_1+\delta_2+\beta)$-marginally differentially
  private, IR with certainty, and with probability $1-\beta$ outputs
  an allocation which is $\alpha$-PO for
  \[\alpha  =  \Oh\left( \frac{k^{3}}{\epsilon n}\cdot \left(\sqrt{k\log(1/\delta_1)}\log(k^3/\beta) +k\sqrt{k\log(1/\delta_2)}\right) \right).   \]
  \end{theorem}

\subsection{Pareto Optimality of $\pttc$}
  We now prove several lemmas about \pttc, from which
  Theorem~\ref{thm:mdp} follows.  We first focus on proving that $\pttc$ is IR and asymptotically Pareto optimal.
\begin{lemma}\label{lem:ir}
The allocation that \pttc produces is IR.
\end{lemma}
\proof Note that in the event that the algorithm fails, then all
agents are allocated the good they are endowed with, which is
trivially IR.  Now suppose that the algorithm runs to completion and at some cycle, agent $i$ was selected to trade.  In order for $i$ to be selected, that means that his good is still available to trade.  Each round, $i$ is on the edge that points to his most preferred good that is available.  With his good $g_i$ still available, $i$ will never be on an arc that points to a less preferred good than $g_i$ and hence can never be selected to trade a less preferred good than $g_i$.
\endproof

We next bound the number of rounds and cycles that are cleared in $\pttc$.

\begin{lemma}
For a fixed round $t$, there can
  be at most $k^2$ cycles selected, i.e. $\tau\leq k^2$ for
  \pttc. Further, the number of rounds is at most $k$.
  \label{lem:running_time}
\end{lemma}
\proof At round $t$, \pttc finds the $\tau$th cycle with noisy weight
$\geq 1$.  Then there is some arc $e$ along the cycle that gets
depleted, i.e.
$\hat w_e(t,\tau) \gets \hat w_e(t,\tau-1) - \hat W(t,\tau) < 1 $ and
hence there will never be another cycle in the same round $t$ that
uses this arc $e$.  Thus, each cycle that is cleared in a round $t$
can be charged to a unique arc in the graph.  There are at most $k^2$
arcs in any round.  Hence, we can clear at most $k^2$ cycles before
there are no arcs (and thus cycles) left at round $t$.  Lastly, each
round has some node that is depleted which causes some agents to have
to select another good to trade with. Thus, after $k$ rounds, no more
nodes will exist.
\endproof

The next claim, proved in Appendix \ref{sec:appendix_pttc}, allows us
to assume that the error from the noise we add in $\pttc$ is small
with high probability.
\begin{claim}\label{lowerror}
  (Low Error) With probability $1-\beta$, the noise terms of $\pttc$
satisfy
\[\max_{t \in [k]} \max_{e \in A} |Z_e^t| \leq E = \frac{\log(k^3/\beta)}{\epsilon'}.\]
 Hence, with probability $1-\beta$,
{\small \begin{equation}
E \leq w_e - \hat w_e \leq 3 E.
\label{eq:webound}
\end{equation}}
\end{claim}

We now want to bound the number of agents that may be left at a node
when it has no outgoing arcs with sufficient weight .

\begin{lemma}
  Assuming the condition in \eqref{eq:webound}, if \pttc cannot find a cycle with
  noisy weight at least $1$, then there exists a node $v$ such that $n_v < D$ for $D =
  (3E+1)k$.
\label{lem:delete}
\end{lemma}
\proof
  If there is no cycle of noisy weight at least $1$, then there must
  be some node $v$ such that every outgoing arc has $\hat w_e < 1$.
  From our Low Error Claim, we then know that the exact arc
  weight $w_e < 3E+1$ from \eqref{eq:webound} for every outgoing arc
  of $v$.  Hence we can count how may agents are on the node, $n_v =
  \sum_{e: e=(v,z), z \in V } w_e < k(3E+1) = D$.
    \endproof

It remains to show that the resulting allocation is asymptotically Pareto optimal.

\begin{theorem}
  For $\epsilon,\delta_1,\delta_2,\beta>0$ we have that $\pttc$
  outputs an $\alpha$-PO allocation for
\[  \alpha =  \tilde\Oh \left( \frac{k^{9/2}}{\epsilon n}  \right)\]
   with probability $1-\beta$. Hence, \pttc is asymptotically Pareto optimal.
\label{thm:pareto}
\end{theorem}
\proof
  We refer the reader to the formal description of our algorithm
  (Algorithm \ref{PTTC}) in the appendix.  We write good type $g=1$ to be the first
  good eliminated in $\pttc$, good $g=2$ the second, and so on, where
  the good types $g \in \G=[k]$.  We will compare the allocation $\pi$
  from \pttc with any other allocation $\pi'$ that Pareto dominates
  $\pi$, i.e. $\pi'(i) \succeq_i \pi(i)$ and $\exists j$ such that
  $\pi'(j) \succ \pi(j)$. We will count the number of agents who could
  possibly prefer $\pi'$.   It suffices to bound the quantity: $\Delta = |\{i \in N:
  \pi'(i) \succ_i \pi(i) \}|.$

  At each round, there are a small number of goods that $\pttc$ ignores and just returns them to their owner at the end.  Removing these goods from the exchange prevents feasible exchanges between other agents and the removed goods.  We seek to bound the total number of agents that may get a better trade in $\pi'$ but the goods they were allocated were not part of the exchange when $\pttc$ allocated the good they got in $\pi$.   
  
  Our algorithm deletes a node only at the end
  of a round, when the noisy supply of at least one good falls
  below $1$. Thus, by Claim~\ref{lowerror} and Lemma~\ref{lem:delete},
  there are at most $D= k(3E+1)$ copies of goods, with probability at least $1-\beta$. We condition on this
  bounded error for the remainder of the proof.

  Recall that $S_e(t,\tau)$ is the set of people that traded (or were satisfied) that
  were along arc $e$ at round $t$ when the $\tau$th cycle $C(t,\tau)$
  was selected in \pttc.  We define the set $S(t)$ to be all the
  people that traded at round $t$: $S(t) = \cup_{\tau \in [k^2]}
  \cup_{e \in C(t,\tau)} S_e(t,\tau).$

  Some agents are not cleared at any round of \pttc, and these agents
  receive the good they were endowed with.  We refer to those people
  that were never selected as $S(n) = N \backslash \{\cup_{t=1}^k S(t)
  \}$.  We then have a partition of $N = \cup_{t=1}^kS(t) \cup S(n)$.

  We now denote the quantity $\Delta_r = |\{i \in S(r): \pi'(i)
  \succ_i \pi(i) \}|$ for $r = 1, \cdots, k, n$: note that $\sum_{r
    \in \{1,\ldots,k, n\}} \Delta_r = \Delta$. Further, we would like
  to refer to the number of goods of type $g$ that were allocated to agents in
  $S(r)$ in $\pi'$ but not in $\pi$, which we define as $\Delta_r(g)$.
  Note that for $g\geq r$ we know that $\Delta_r(g) = 0$ because
  agents cleared by \pttc receive their favorite good among the one's
  remaining at the round in which they are cleared, and all goods
  $g\geq r$ are available at round $r$.  More formally we have:
  \begingroup
\everymath{\scriptstyle}
\[\Delta_r(g) = |\{ i \in S(r) : \pi'(i) = g \succ_i \pi(i) \}| \text{ for } r = 1, \cdots, k, \quad g <r.\]
\endgroup
Thus we can write: $\Delta_r = \sum_{g = 1}^{r-1} \Delta_r(g)$. Note
that at the first round $r=1$ that $\Delta_1 = 0$ because everyone
that was selected gets their favorite type of good.  Let us also
define $\Delta_r(g,h)$ to be the number of goods of type $g$ allocated
to agents in $S(r)$ in $\pi'$ who received good $h$ in $\pi$:
\begingroup
\everymath{\scriptstyle}
\[\Delta_r(g,h) = |\{ i \in S(r) : \pi'(i) = g \succ_i h = \pi(i) \}| \quad \text{ for } r = 1, \cdots, k, \quad h \geq r \quad g <r.\]
\[\Delta_r(g)=\sum_{h = r}^k \Delta_r(g,h) \qquad g<r\]
\endgroup
We denote the number of initial goods of type $g$ as $n(g)$ for $g \in
[k]$: i.e. $n(g) = |{i : g_i = g}|$.   We define $n_t(g)$ as the number
of goods of type $g$ that are not allocated to members of $\cup_{r =
  1}^t S(r)$ in $\pi'$ (our notation uses the round as a subscript
and the good as an argument in parentheses), i.e.
\begingroup
\everymath{\scriptstyle}
\[n_t(g) = n(g) - \sum_{r = 1}^t |\{ i \in S(r): \pi'(i) = g \}|.\]
\endgroup
We will now bound the quantities $n_t(g)$.  Note that for the first round, we have
\begingroup
\everymath{\scriptstyle}
\[n_1(1) = n(1) - |\{ i \in S(1): \pi'(i) = 1 |\leq D \]
\endgroup
because each person in $S(1)$ got his favorite good in \pttc, and
since $\pi'$ Pareto dominates $\pi$, $\pi'$ must have made the same
allocation as $\pi$ for $S(1)$, except for the at most $D$ agents
\pttc never selected that had a good of type 1. The agents that get
selected later can Pareto improve because they could have selected
good type 1, but \pttc has deleted that good for future rounds. These
lost copies of good 1 can potentially be used in allocation $\pi'$
to improve the outcome of agents selected at future rounds by
\pttc. We will account for these improvements by keeping track of the
 $n_r(1)$ copies of good 1 remaining at each round, where
 \begingroup
\everymath{\scriptstyle}
\[n_r(1) = n_{r-1}(1) - \Delta_r(1) \qquad \text{ for } r = 2, \cdots, k.\]
\endgroup
We then continue in this fashion with good type 2 in order to bound
$n_2(2)$.  We know that $\pi'$ allocates $\Delta_2(1)$ goods of type
$1$ to agents in $S(2)$.  Because $\pi'$ Pareto dominates $\pi$, it
must be that $\pi'$ makes the same allocations as $\pi$ among agents
in $S(2)$, except for the agents that $\pi'$ matches to good type 1
that got good type 2 in $\pi$ (these are the only agents who are
possibly not getting their favorite good among those ``remaining'' in
$\pi'$).  We then bound the number of goods of type 2 that have not been distributed to people in $S(1)$ or $S(2)$ and then keep track of the number of these goods that remain to give to the people that are selected in future rounds but get good type 2 in $\pi'$,
\begingroup
\everymath{\scriptstyle}
\[n_2(2) \leq D + \Delta_2(1,2)\]
\[n_r(2) = n_{r-1}(2) - \Delta_r(2) \qquad r > 2.\]
\endgroup
We now consider $n_3(3)$, the number of goods of type $3$ that $\pi'$
has not allocated to members in $S(1), S(2)$, and $S(3)$.  This is the
same as the number of goods of type 3 that $\pi$ will never give to
selected people (at most $D$) in addition to the goods of type 3 that $\pi$ gave to people in $S(2)$ and $S(3)$ that $\pi'$ gave a
different good to, i.e. the $ \Delta_2(1,3)$ people that got good 3 at
round 2 in $\pi$, but $\pi'$ gave them good 1, along with the
$\Delta_3(1,3)$ and $\Delta_3(2,3)$ people that got good 3 at round 3
in $\pi$, but $\pi'$ gave them good 1 and 2 respectively.  This  implies both
\begingroup
\everymath{\scriptstyle}
\[n_3(3) \leq D + \Delta_2(1,3) + \Delta_3(1,3) + \Delta_3(2,3)\]
\[n_r(2) = n_{r-1}(3) - \Delta_r(3) \qquad r > 3.\]
\endgroup
We then generalize the relation for $r \geq 3$:
\begingroup
\everymath{\scriptstyle}
\begin{equation}
n_r(r) \leq D + \sum_{\ell = 2}^{r} \sum_{g = 1}^{\ell-1} \Delta_\ell(g,r)
\label{eq:nrr}
\end{equation}
\begin{equation}
n_t (r) = n_{t-1}(r) - \Delta_t(r) \qquad t>r.
\label{eq:ntr}
\end{equation}
\endgroup
Because the number of goods remaining at each round must be
nonnegative, we have
\begingroup
\everymath{\scriptstyle}
\begin{equation}
\Delta_t(r) \leq n_{t-1}(r).
\label{eq:nonneg}
\end{equation}
\endgroup
Recall that $\Delta_1 = 0$.  We also have:
$
\Delta_2 = \Delta_2(1) \leq n_1(1) \leq D.
$
For round $3\leq t\leq k$, we use \eqref{eq:nrr}, \eqref{eq:ntr}, and
\eqref{eq:nonneg} to get:
\begingroup
\everymath{\scriptstyle}
\begin{align*}
\Delta_t & = \sum_{g = 1}^{t-1} \Delta_t(g) \underbrace{\leq}_{\eqref{eq:nonneg}} \sum_{g = 1}^{t-1} n_{t-1}(g)  \underbrace{=}_{\eqref{eq:ntr}} \sum_{r= 1}^{t-1} n_r(r) - \sum_{g=1}^{t-2} \sum_{r=g+1}^{t-1} \Delta_{r}(g)   \\
& \underbrace{\leq}_{\eqref{eq:nrr}} (t-1)D + \sum_{r=2}^{t-1} \sum_{\ell = 2}^{r} \sum_{g = 1}^{\ell-1} \Delta_\ell(g,r) - \sum_{g=1}^{t-2} \sum_{r=g+1}^{t-1} \Delta_{r}(g) \\
& = (t-1)D + \sum_{\ell = 2}^{t-1} \sum_{g=1}^{\ell-1} \underbrace{ \sum_{r = \ell}^{t-1} \Delta_\ell(g,r)}_{\leq \Delta_\ell(g) } - \sum_{r = 2}^{t-1} \sum_{g = 1}^{r-1} \Delta_r(g) \leq (t-1)D.
\end{align*}
\endgroup
We next bound $\Delta_n$.  All agents who were never cleared might be
able to get a better good in $\pi'$, which can not be more than the total number of 
goods that $\pi'$ did not allocate to any of the selected people.
Thus we have
\begingroup
\everymath{\scriptstyle}
\begin{align*}
\Delta_n & \leq \sum_{g = 1}^k n_k(g) = \sum_{r = 1}^k n_r(r) - \sum_{g = 1}^{k-1} \sum_{r = g+1}^{k} \Delta_r(g) \\
& \leq kD  + \sum_{\ell = 2}^{k} \sum_{g=1}^{\ell-1} \underbrace{\sum_{r =\ell}^{k} \Delta_\ell(g,r)}_{ = \Delta_\ell(g)} - \sum_{r = 2}^{k} \sum_{g = 1}^{r-1} \Delta_r(g) = kD.
\end{align*}
\endgroup
We then sum over $\Delta_t$ for every round $t$ to get $\Delta$:
\begingroup
\everymath{\scriptstyle}
\begin{align*}
  \Delta_1& +  \Delta_2 + \Delta_3 + \cdots + \Delta_k + \Delta_n  \leq \sum_{j = 1}^k j D = \Oh\left( k^2 D \right)  \\
          & = \Oh\left( \frac{k^{3} \left(k\sqrt{k\log(1/\delta_2)}+\log(k^3/\beta)\sqrt{k\log(1/\delta_1)}\right)}{\epsilon} \right)
\end{align*}
\endgroup
where the last equality followed from $D \leq k(3E+1)$, from
$\epsilon'$ in \eqref{eq:eps_prime}, and $E$ in \eqref{eq:error}.
\endproof

\subsection{Privacy Analysis of $\pttc$}
    In order to prove that $\pttc$ is marginally differentially
    private, we first present some known results on differential
    privacy and prove some lemmas that will be useful.  We first state the definition of
    sensitivity of a function, which will help when we define the
    Laplace Mechanism \citep{DMNS06} which is a subroutine of \pttc.
\begin{definition}[Sensitivity \citep{DMNS06}]
  A function $\phi: \cX^n \to \R^m$ has \emph{sensitivity} $\nu_\phi$ defined as
\[\nu_{\phi} =\max_{i \in [n] } \max_{\bbx_{-i}\in \cX^{n-1}}\max_{ x_i \neq x_i'  } || \phi(\bbx_{-i}, x_i ) - \phi(\bbx_{-i}, x_i') ||_{1}.\]
\end{definition}
The Laplace Mechanism, given in Algorithm \ref{Lap_Mech}, answers a
numeric query $\phi: T^n \to \R^m$ by adding noise to each component
of $\phi$'s output in a way which is differentially private.

\ifnum\final=0
\begin{algorithm}
\caption{Laplace Mechanism}\label{Lap_Mech}
\begin{algorithmic}[0]
\INPUT : Database $\bbx$.
\OUTPUT : An approximate value for $\phi$
\Procedure {$M_L$} {$\epsilon, g$}$(\bbx)$
\State $\hat \phi = \phi(\bbx) + (Z_1, \cdots, Z_m) \qquad Z_i \stackrel{i.i.d.}{\sim} $ Lap$(\delta(\phi)/\epsilon)$
\State {\bf return} $\hat \phi$.
\EndProcedure
\end{algorithmic}
\end{algorithm}

\else
\begin{algorithm}[h]
\SetAlgoNoLine
\KwIn{Database $\bbx$}
\KwOut{An approximate value for $\phi$}
{ $\qquad \hat \phi= \phi(\bbx) + (Z_1, \cdots, Z_m) \qquad Z_i \stackrel{i.i.d.}{\sim}  \text{ Lap}(\nu_{\phi}/\epsilon)$\\
 \Return $\hat\phi$
}
\caption{Laplace Mechanism $M_L (\epsilon,\phi) (\bbx)$}
\label{Lap_Mech}
\end{algorithm}
\fi


We now state the privacy guarantee for the Laplace Mechanism $M_L$.
\begin{theorem}[\citep{DMNS06}]
  $M_L(\epsilon, \phi)$ is $\epsilon$-differentially private for any
  $\phi: \cX^n \to \R^n$ with bounded sensitivity.
  \label{thm:lap_priv}
\end{theorem}
Our algorithm \pttc uses the Laplace Mechanism to
modify arc weights at each round.

One of the most useful properties given by differential privacy is its
ability to compose: running a collection of private mechanisms is also
private, with a loss in the privacy parameters which depends upon how
the composition is done.  We will need to use two composition
theorems.  The first shows that the privacy parameters add when we
compose two differentially private mechanisms, and the second from
\cite{DRV10} gives a better composition guarantee even with many
adaptively chosen mechanisms.
\begin{theorem}
  If mechanism $M_1: \mathcal{X}^n \to O$  is
  $(\epsilon_1, \delta_1)$-differentially private, and another
  mechanism $M_2:\mathcal{X}^n \times O \to R$ is
  $(\epsilon_2,\delta_2)$-differentially private in its first
  component, then $M: \mathcal{X}^n \to R$ is
  $(\epsilon_1+\epsilon_2,\delta_1+\delta_2)$ differentially private
  where
$M(\mathbf{x}) = M_2(\mathbf{x},M_1(\mathbf{x})).$
\label{thm:comp}
\end{theorem}
If we were to only consider the previous composition theorem, then the
composition of $m$ mechanisms that are
$(\epsilon,\delta)$-differentially private mechanisms would lead to an
$(m\epsilon,m\delta)$-differentially private mechanism.  However, we
can improve on the $m\epsilon$ parameter at the cost of making the
additive $m\delta$ term larger.  The following theorem gives this
modified composition guarantee that holds even when the sequence of
differentially private mechanisms is chosen adaptively by an
adversary, as a function of the output of previous mechanisms.
 
\begin{theorem}[$m$-Fold Adaptive Composition \citep{DRV10}]
  Fix $\delta>0$.  The class of $(\epsilon',\delta')$ differentially
  private mechanisms satisfies $(\epsilon,m\delta' + \delta)$
  differential privacy under $m$-fold adaptive composition for
\[\epsilon ' = \frac{\epsilon}{\sqrt{8 m \log(1/\delta)}}.\]
\label{thm:advanced_comp}
\end{theorem}

We next give a lemma that will be useful in proving marginal
differential privacy, given an intermediate step that is
differentially private.

\begin{lemma}
  Let $M^1: \cX^n \to O$ be $(\epsilon_1,\delta_1)$-differentially
  private and $M_j^2: \cX^{n-1}\times \cX \times O \to R$ for
  $j = 1, \cdots, n$ be $(\epsilon_2,\delta_2)$- differentially
  private in its first argument.  Then $M: \cX^n \to R^n$ is
  $(\epsilon_1+\epsilon_2,\delta_1+\delta_2)$-marginally
  differentially private where
\[M(\bbx) = (M_j^2(\bbx_{-j},x_j,M^1(\bbx)))_{j=1}^n.\]
\label{lem:mdp}
\end{lemma}
\begin{proof}
  To prove marginal differential privacy we need to prove that the
  component $M(\bbx_{-j},x_j)_j$ is differentially private in its
  first argument, for every $j$.  Fix any index $i$ and $j$ such that
  $i\neq j$.  We then consider a $\bbx_{-i} \in \cX^{n-1}$,
  $x_i' \neq x_i$, and $S\subset R$, then we use composition of differentially private mechanisms to get the
  following:
{ \begin{align*}
&\Prob\left(M_j^2[\bbx_{-j},x_j, M^1(\bbx) ] \in S\right)
 = \int_{ O} \Prob( M_j^2( (\bbx_{-(i,j)},x_i),x_j,o )\in S ) \Prob(M^1(\bbx_{-i},x_i) = o) d o \\
 &\leq  \int_{ O} \Prob( M_j^2( (\bbx_{-(i,j)},x_i ),x_j,o )\in S )\left(e^{\epsilon_1} \Prob(M^1(\bbx_{-i},x'_i) = o) +\delta_1\right)d o \\
& \leq \int_{ O} \left(e^{\epsilon_2}\Prob( M_j^2( (\bbx_{-(i,j)},x'_i ),x_j,o )\in S ) +\delta_2\right)e^{\epsilon_1} \Prob(M^1(\bbx_{-i},x'_i) = o) d o  +\delta_1\\
& \leq e^{\epsilon_1+\epsilon_2} \Prob(M_j^2[ (\bbx_{-(i,j)},x'_i ),x_j , M^1(\bbx_{-i},x'_i) ] \in S) +\delta_1+\delta_2
        \end{align*}}
      Where in the last inequality we use the fact that if $e^{\epsilon_1}\Prob(M^1(\bbx_{-i},x'_i) = o) >1$ then we just replace the bound by $1$.\end{proof}

We now present two lemmas specific to our setting: one dealing with
the privacy of the noisy arc weights computed in $\pttc$ and the
other dealing with how agents are selected to trade.  We leave both
proofs to Appendix \ref{sec:appendix_pttc}.
\begin{lemma}\label{lem:mdp_lem1}
  The mechanism $\cM: \cX^n \to \R^{k\times k^2}$ that outputs all the
  noisy arc weights $(\hat w(t,0))_{t \in [k]}$ that is used in
  $\pttc $ is $(\epsilon_1,\delta_1)$-differentially private for
  $\delta_1>0$ and
\begin{equation}
\epsilon_1 = 2\epsilon'\cdot\sqrt{8k\log(1/\delta_1)}.
\label{eq:eps_1}
\end{equation}
\end{lemma}
\proof[\ifnum\final=0 Proof \fi Sketch] At each round of $\pttc$, we add Laplace noise with parameter
$1/\epsilon'$ to each arc weight.  Consider agent $i$ changing her input from
$x_i$ to $x_i'$ and fix all the other agents' inputs.  We want to bound the
sensitivity of the (exact) arc weights at any particular round.  If we fix
the randomness used in $\rselect$ that determines which agents are cleared,
and the previous rounds' noisy arc weights, then we can consider several
cases in how agent $i$ changing reports can affect the arc counts at the
next round. We show that no matter whether agent $i$ is selected earlier or
later on a different input, she can affect at most 2 entries in the vector of
arc weights at a particular round by at most 1. Hence, the sensitivity of the arc weights at a particular round is at most 2,
conditioned on the past rounds' randomness.  We then apply advanced
composition in Theorem \ref{thm:advanced_comp} over $k$ rounds to get
$\epsilon_1$.
\endproof
For the next lemma, we will condition on knowing the noisy arc weights that
are computed throughout $\pttc$ and compute the difference in distributions
between agent $j$ being selected at any round when agent $i\neq j$ changes
reports from $x_i$ to $x_i'$.

\begin{lemma}\label{lem:mdp_lem2}
  Let $\hat w$ be the noisy arc weights that are computed for every
  round of our algorithm $\pttc$.  We also assume that each entry of
  $\hat w$ satisfies \eqref{eq:webound}.  The mechanism
  $\cM_j:\cX^{n-1} \times \cX \times \R^{k \times k^2} \to \{0,\G\}$
  is $(\epsilon_2,\delta_2)$-differentially private in its first
  argument given $\hat w$, where
\[\cM_j(\bbx_{-j},x_j,\hat w) = \left\{\begin{array}{ll}
				g &\qquad  \text{If } \exists t\in [k],\tau\in [k^2] \text{ s.t. } j \in S_e(t,\tau) \text{ for } e=(g_j,g) \in C(t,\tau) \\
				0 &\qquad  \text{else}
				\end{array}\right..\]
with $\delta_2>0$ and
\begin{equation}
 \epsilon_2 = \frac{2k\sqrt{8k\log(1/\delta_2)}}{E} \label{eq:eps_2}.
 \end{equation}
\end{lemma}
\proof[\ifnum\final=0 Proof \fi Sketch] We consider a round $t$ and cycle $C$ such that agent $j$ is
on an arc $e \in C$ and so may or may not be selected at that cycle.  We then
bound the ratio between the probability that agent $j$ is selected when $i
\neq j$ reports $x_i$ to the probability of the same event when $i$'s input
is $x_i'$. We are given the number of people that are being selected (this is
a function of the noisy arc weights).  As we argued in the previous lemma, agent
$i$'s influence on the weight of any arc at any round by at most 1.
  We do a similar analysis when the event is $j$ not being selected.  Knowing this, and the
fact that the noisy arc weights satisfy \eqref{eq:webound}, we get that the
mechanism that determines whether $j$ is selected or not at a given round, as
part of a given cycle is $2/E$-differentially private in reports $\bbx_{-j}$.

We then apply advanced composition (Theorem \ref{thm:advanced_comp}) over the
at most $k^2$ different cycles that might need to be cleared on a single
round, and over the $k$ different rounds, to get the value for $\epsilon_2$.
\endproof
We are now ready to prove that $\pttc$ is marginally differentially private.
\begin{theorem}
For parameters $\epsilon,\delta_1,\delta_2,\beta>0$, $\pttc:\cX^n \to \G^n$ is $(\epsilon,\delta_1+\delta_2+\beta)$-marginally differentially private.
\label{thm:pttcprivate}
\end{theorem}
\begin{proof}
  We fix agents $i$ and $j$, where $i \neq j$.  Let us define $M^1$ as
  $\cM$ that outputs all the noisy arc weights for each round of
  $\pttc$ from Lemma \ref{lem:mdp_lem1} and $M^2_j$ as $\cM_j$ the
  good that $j$ ends up being matched with (or zero if never matched)
  from Lemma \ref{lem:mdp_lem2}.

  We first condition on the event that all the noisy arc weights
  computed by $M^1$ satisfy \eqref{eq:webound}.  We then apply Lemma
  \ref{lem:mdp} to get the composed mechanism $M:\cX^n \to \{0,\G\}^n$
  of our mechanisms $M^1$ and $(M^2_j)_{j=1}^n$ is
  $(\epsilon_1+\epsilon_2,\delta_1+\delta_2)$-marginally
  differentially private.  Note that $M$ and $\pttc$ have the same
  distribution of outcomes where if $M$ outputs $0$ to agent $j$,
  then we know $\pttc$ will give agent $j$ his own good type $g_j$.

  We have yet to consider the case when the noisy arc weights do not
  satisfy \eqref{eq:webound}.  However, this occurs with probability
  at most $\beta$.  We conclude then that $M$ (and thus $\pttc$) is
  $(\epsilon_1+\epsilon_2,\delta_1+\delta_2+\beta)$-marginally
  differentially private.  We then plug in the values for $\epsilon_1$
  in \eqref{eq:eps_1}, $\epsilon_2$ in \eqref{eq:eps_2}, $\epsilon'$
  in \eqref{eq:eps_prime}, and $E$ in \eqref{eq:error} to get
 \[ \epsilon_1 + \epsilon_2 = \frac{2\epsilon'\sqrt{8k}}{\log(k^3/\beta)}\cdot\left( \log(k^3/\beta)\sqrt{\log(1/\delta_1)} + k \sqrt{\log(1/\delta_2)} \right) = \epsilon\]
\end{proof}

We want $\delta_1,\delta_2$ and $\beta$ to be as small as possible
because this causes an additive difference in the probability
distributions between neighboring exchange markets.  We can then set
$\delta_1,\delta_2, \beta = \text{poly}(1/n)$ to still
get $\alpha = \tilde \Oh\left( \frac{k^{9/2}}{\epsilon n} \right).$

\section{Allowing a Small Supply of Goods to be Injected}
Without privacy
constraints, rather than running the top trading cycles algorithm, we could
solve the following linear program to obtain an IR and
PO exchange:
\begin{align}
\max_{\mathbf{z}} & \qquad \sum_{i \in [n]}\sum_{j \in \G} r_{ij} z_{ij} \label{eq:lp_match} \\
s.t. & \qquad \sum_{j \in \G} z_{ij} = n_j \qquad \forall i \in [n] \label{eq:LP_alloc} \\
& \qquad \forall j \in \G \text{ s.t. } g_i \succ_i j \qquad z_{i,j} = 0\qquad \forall i \in [n] \label{eq:LP_IR} \\
& \qquad \sum_{i \in [n]} z_{ij} = 1 \qquad \forall j \in \G \label{eq:LP_match} \\
& \qquad z_{ij} \in \{0,1\}. \label{eq:LP_int}
\end{align}
where we define $r_{ij} = k - r+1$ if $j \succeq_i g_i$ and $j$ is the $r$th
top choice of agent $i$.  If $g_i \succ_i j$ then $w_{ij} = 0$. Note that
since this is just a max-weight matching problem, the optimal solution
$\mathbf{x}^*$ will be integral, even if we relax the integer constraint. The
constraint in \eqref{eq:LP_IR} ensures that no agent is matched to a good
that is preferred less than her endowed good, which enforces IR. Finally, the optimal solution to this linear program is an
(exactly) Pareto optimal allocation: any other allocation that gave some
agents strictly preferred goods, without decreasing the quality of goods
given to other agents would have a strictly improved objective value,
contradicting optimality.  We thus have the following theorem:
\begin{theorem}
  A solution to the integer program (IP) in \eqref{eq:lp_match} - \eqref{eq:LP_int} is
  an IR and PO allocation.
\end{theorem}

We leverage recent work \cite{Matching} and \cite{HHRW14} on computing
max-weight matchings subject to joint differential privacy.  We will
use the results of the latter paper to conclude that we can obtain an
asymptotically Pareto optimal, IR allocation that is also jointly
differentially private under a relaxation of our problem that allows
us to inject a small number of extra goods into the system.  We
present the following theorem in the context of the IP of
\eqref{eq:lp_match} - \eqref{eq:LP_int}.
\begin{theorem}[\cite{HHRW14}]
Let $OPT$ be the optimal objective value of \eqref{eq:lp_match}.  For
$\epsilon,\delta,\beta>0$ there exists an $(\epsilon,\delta)$-jointly
differentially private algorithm that produces fractional solution
$\bar{\mathbf{z}}$ such that $\sum_{j \in \G} \bar{z}_{ij} = 1$ and with
probability $1-\beta$
\begin{itemize}
\item[$\bullet$] We get a solution close to $OPT$, i.e. $\sum_{i \in [n]}
    \sum_{j \in \G} r_{ij} \bar{z}_{ij} \geq OPT - \eta$ where
\[\eta = \Oh\left( \frac{k^2 \log(nk/\beta)\log^{1/2}(n/\delta)}{\epsilon} \right)\]
\item[$\bullet$] The total amount all the constraints in
    \eqref{eq:LP_match} are violated by $ \bar{\mathbf{z}}$ is small, i.e.
    we have for each good type $j \in \G$,  $n_j - \lambda_j \leq \sum_{i
    \in [n]} \bar{z}_{ij} \leq n_j + \lambda_j$ where
\[\sum_{j \in \G} \lambda_j = \Oh\left( \frac{k  \log(nk/\beta)\log^{1/2}(n/\delta)}{\epsilon}\right)\]
\end{itemize}
\end{theorem}

Note that the solution $\bar{\mathbf{z}}$ we obtain is fractional.  We
then use a randomized rounding technique from \citet{RT87} that has
each agent $i$ choose good $j$ with probability $\bar{z}_{ij}$.  By
applying a Chernoff bound, the resulting solution after applying
randomized rounding $\hat{z}_{ij}$ gives with probability at least
$1-\beta$
\begin{itemize}
\item[$\bullet$]
It is the case that  $\sum_{i \in [n]} \sum_{j \in \G} r_{ij} \hat{z}_{ij} \geq OPT -
  \hat\eta $ where
\[\hat\eta = \Oh\left( \frac{k^2\sqrt{n} \log^{3/2}(nk/\beta)  \log^{1/2}(n/\delta)}{\epsilon} \right)\]
\item[$\bullet$] The total amount all the constraints in \eqref{eq:LP_match} are violated by $\hat{\mathbf{z}}$ is small, i.e. $n_j - \hat\lambda_j \leq \sum_{i \in [n]} \hat{z}_{ij} \leq n_j + \hat\lambda_j$ where
\[
\sum_{j \in \G} \hat\lambda_j = \Oh\left( \frac{k \sqrt{n} \log^{3/2}(nk/\beta) \log^{1/2}(n/\delta)}{\epsilon}\right)
\]
\end{itemize}
Since the supply constraints are violated in the above solution, it is
infeasible, and cannot be implemented if the market is closed. Moreover, this
is inherent -- everything we have done here is subject to joint differential
privacy, for which we have proven a lower bound. However, if we have a supply
of \emph{extra} goods of each type (e.g. non-living kidney donors), then we
can use these extra goods to restore feasibility (Note that it is important
that the ``extra'' goods are not attached to agents who have IR constraints).
This allows us to circumvent our lower bound, and leads to the following
theorem:

\begin{theorem}
There exists an $(\epsilon, \delta)$ joint differentially private algorithm
that allocates goods in an exchange market that is $\alpha$ - PO with
probability $1-\beta$ and always IR which needs at most a total of $\Lambda$
extra goods to ensure everyone gets a good where
\[
\alpha = \Oh\left( \frac{k^2\log^{3/2}(nk/\beta) \log^{1/2}(n/\delta) }{\sqrt{n}\epsilon} \right) \quad\text{and}\quad \Lambda = \Oh\left( \frac{k \sqrt{n}  \log^{3/2}(nk/\beta) \log^{1/2}(n/\delta)}{\epsilon}\right)
\]
\end{theorem}

\section{Conclusion/Open Problems}

In this paper we have continued the study of the accuracy to which \emph{allocation problems} can be solved under parameterized relaxations of \emph{differential privacy}. Generically, these kinds of problems cannot be solved under the standard constraint of differential privacy. Unlike two sided allocation problems which can be solved under \emph{joint}-differential privacy, we show that Pareto optimal exchanges cannot be solved even under this relaxation, but can be solved asymptotically exactly under marginal differential privacy whenever the number of types of goods $k = o(n^{2/9})$. (We note that in many applications, such as kidney exchange, $k$ will be constant).

The two privacy solution concepts we have considered are only two extremes along a spectrum: informally, joint differential privacy protects the privacy of agent $i$ against an adversarial collusion of possibly all of the $n-1$ other agents in the market, acting against agent $i$. Similarly, marginal differential privacy protects agent $i$'s privacy only against a single agent $j \neq i$ in the computation, assuming she does not collude with anyone else. We propose a definition for future work which smoothly interpolates between joint and marginal differential privacy, which we call $m$-coalition differential privacy. The case of $m=1$ recovers marginal differential privacy and the case of $m=n-1$ recovers joint differential privacy: for $1 < m < n-1$, we get a sequence of privacy definitions smoothly interpolating between the two.
\begin{definition}[Coalition Differential Privacy]
We say that a mechanism $M: \mathcal{X}^n \to A^n$ is $m$-\emph{coalition} $(\epsilon,\delta)$-\emph{differentially private} if for any set $S \subseteq [n]$ with $|S|\leq m$ and for any $\mathbf{x} = (x_1,\cdots, x_n)$ and $x_i' \neq x_i$ where $i \notin S$ we have for any $B \subset O^{|S|}$
$$
\Prob\left( M(\mathbf{x}_{-i},x_i)_S \in B \right) \leq e^{\epsilon} \Prob\left( M(\mathbf{x}_{-i},x'_i)_S \in B\right) + \delta
$$
where $M(\mathbf{x})_S = (M(\mathbf{x})_j)_{j \in S}$
\end{definition}
We note that it is not true in general that $\epsilon$-marginal differential privacy implies to $m$-coalition $(m\epsilon)$-differential privacy, because the marginal distributions between players may be correlated, and so this study may require new tools and techniques.

It would also be interesting to give a privacy preserving algorithm that is not only individually rational and asymptotically Pareto optimal, but makes truthful reporting a dominant strategy. One difficulty is that in our setting (in which there are multiple copies of identical goods), agents do not have strict preferences over goods, and even the top trading cycles algorithm without privacy is not incentive compatible. However, there are other algorithms such as \citep{SS13,AM09,JM12}, that are incentive compatible in exchange markets that allow indifferences, so it may be possible. (It would also be interesting to find a connection between marginal differential privacy and incentive compatibility, like the known connections with differential privacy \cite{MT07} and joint differential privacy \cite{Large}).

\bibliographystyle{acmsmall}
\bibliography{refs_ttc}
\newpage

\appendix

\section{Formal Proofs from Lower Bounds Section \ref{SEC:LB}}\label{section:lb_proofs}
\proof[\ifnum\final=0 Proof \fi of Claim \ref{claim:bit}]
  Let $M$ be $(\epsilon,\delta)$-differentially private such that $\Prob(M(0) =
  0)> \frac{e^\epsilon+\delta}{e^\epsilon + 1}$ and $\Prob(M(1) = 1) >
  \frac{e^\epsilon+\delta}{e^\epsilon + 1}$.  By the definition of
  $\epsilon$-differential privacy, we have

\begin{align*}
\frac{e^\epsilon+\delta}{e^\epsilon + 1}& < \Prob(M(0) = 0) \leq e^\epsilon \Prob(M(1) = 0)+\delta\\
&  = e^\epsilon(1-\Prob(M(1) = 1)) +\delta< e^\epsilon \left(1-\frac{e^\epsilon+\delta}{e^\epsilon + 1}\right)+\delta,
\end{align*}
 giving a contradiction.
\endproof

\begin{proof}[\ifnum\final=0 Proof \fi of Theorem \ref{thm:lowerbound}]
  The proof proceeds by a reduction to Claim~\ref{claim:bit}. Suppose we
  had such an $(\epsilon,\delta)$-joint differentially private mechanism
  $M_J$, with $\alpha < 1-
  \frac{e^\epsilon+\delta}{(1-\beta)(e^\epsilon+1)}$. We show that we could
  use it to construct an $\epsilon$-differentially private mechanism
  $M: \{0,1 \} \to \{ 0,1\}$ that contradicts Claim~\ref{claim:bit}.
  We design an exchange market parameterized by the input bit $b$
  received by mechanism $M$ (see Figure~\ref{fig:thresh}). The market
  has two types of goods, $g_0$ and $g_1$ and $2n$ agents partitioned
  into two sets, $N_0$ and $N_1$ of size $n$ each. The agents $j \in
  N_1$ are endowed with good $g_1$ and strictly prefer good $g_0$
  (i.e. all such agents $j$ have preference $g_0 \succ_j g_1$). The
  agents in $N_0$ are endowed with good $g_0$. We assume $n-1$ of them strictly
  prefer good $g_1$ (i.e. all such agents $j$ have preference $g_1
  \succ_j g_0$). A distinguished agent, $i \in N_0$, selected among
  the $n$ agents in $N_0$ uniformly at random, has preference
  determined by bit $b$: $g_b \succ_i g_{1-b}$. (i.e. the $i$'th agent
  wishes to trade if $b = 1$, but prefers keeping her own good if $b =
  0$.)  We denote the vector of linear preferences that depends on $i$'s bit $b$ as $\succ(b)$.  We will refer to this exchange market as $\mathbf{x}(b) =
  (\mathbf{g},\succ(b) )\in \mathcal{X}^{2n}$ where $\mathbf{g} \in
  \{0,1 \}^{2n}$. We remark that when $b = 1$, the agents
  in $N_0$ are identical. \\
\begin{figure}[h]
\centerline{\includegraphics[width=4.5in]{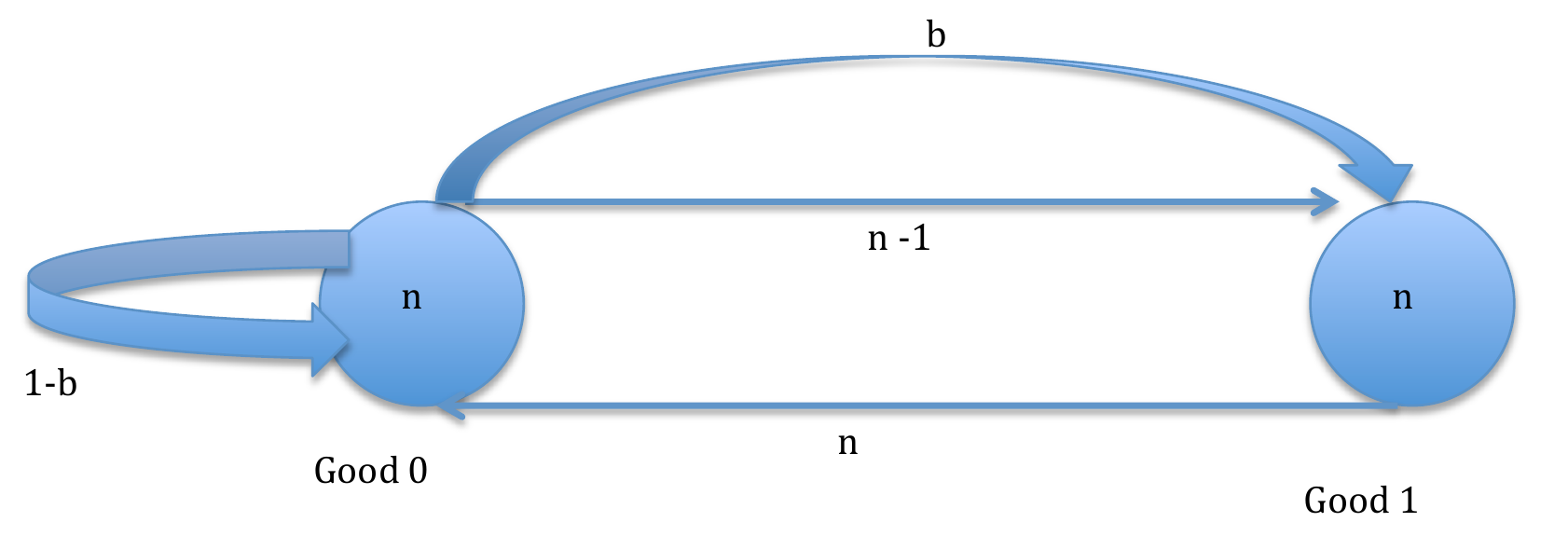}}
\caption{Depicting the exchange market considered in the proof of Theorem \ref{thm:lowerbound}.}
\label{fig:thresh}
\end{figure}

Let $M_J: \mathcal{X}^{2n} \to \G^{2n}$ be the $(\epsilon,\delta)$-joint
differentially private mechanism given in the statement of the theorem.
Note first that by the definition of joint differential privacy, the
mechanism $M': \{0,1\} \to \G^{n-1}$ defined as $M'(b) =
M_J(\mathbf{x}(b))_{-i}$ (which takes as input $b$ and outputs the
allocation of all $n-1$ agents $j \neq i$) is
$(\epsilon,\delta)$-differentially private.

In this construction, when $b = 0$, the IR constraint
requires that $M_J(\mathbf{x}(0))_i =0$ with probability $1$. Note
also that from the output of $M'(b)$, we can determine whether agent
$i$ engaged in trade: if not, we must have:
$$|\{j \in N_0\backslash i : M'(b)_j = g_1\}| = |{j \in N_1 : M'(b)_j = g_0}|$$
but if so, we must have:
$$|\{j \in N_0\backslash i : M'(b)_j = g_1\}| = |{j \in N_1 : M'(b)_j = g_0}| - 1$$

Define $f:\G^{2n-1}\rightarrow \{0,1\}$ to be the indicator function
of the event $|\{j \in N_0\backslash i : M'(b)_j = g_1\}| \neq |{j \in N_1 :
  M'(b)_j = g_0}|$. Define $M:\{0,1\}\rightarrow \{0,1\}$ to be $M(b)
= f(M'(b))$. Note that $M$ is $(\epsilon,\delta)$-differentially private by the
post-processing guarantee of differential privacy, and is an indicator
variable determining whether agent $i$ has traded in our exchange
economy.

If $M(b) = 1$, therefore, it must be that $b = 1$ by the individual
rationality guarantee of $M_J$.  (i.e. $\Prob(M(0) = 1) = 0 \implies
\Prob(M(0) = 0) = 1$).  Hence, by Claim~\ref{claim:bit} we must have
$\Prob(M(1) = 1) \leq \frac{e^\epsilon+\delta}{e^\epsilon + 1}$.  We
know by hypothesis that $M_J$ finds an $\alpha$-PO allocation with probability $1-\beta$.  If $b=1$ then every
agent wishes to trade in $\mathbf{x}(b)$ and hence with probability
$1-\beta$, $M_J$ must produce an allocation in which at least
$2n(1-\alpha)$ people trade. Since all agents in $N_1$ are identical,
and $i$ was selected uniformly at random, it must therefore be that
agent $i$ engages in trade with probability at least
$(1-\beta)(1-\alpha)$.  Thus, we have

\[\frac{e^\epsilon+\delta}{e^\epsilon + 1} \geq \Prob(M(1) = 1) \geq (1-\beta)(1-\alpha)\]
which gives us the conclusion of the theorem.
\end{proof}

\begin{proof}[\ifnum\final=0 Proof \fi of Theorem \ref{thm:lb-marginal}]
  Suppose we have an $(\epsilon,\delta)$-marginally differentially
  private mechanism $M_S$ which is IR and $\alpha$-PO. We will use it
  to construct some $(\epsilon,\delta)$-differentially private mechanism $M$, which will
  give a lower bound on $\alpha$ by Claim~\ref{claim:bit}.

  Suppose there are $k$ types of goods in the market, $1, \ldots,
  k$. Let $g_i$ represent the type of good with which agent $i$ is
  initially endowed. Figure \ref{fig:preferences} shows the favorite
  and second favorite goods of all $n$ agents \footnote{Note that,
    since we are only considering IR mechanisms, preferences need only
    be specified to the level where $i$ ranks $g_i$, and each agent
    $i$ in our example has good $g_i$ as her first or second choice.}.
  We will refer to this exchange market as $\mathbf{x}(b) =
  (\mathbf{g},\succ(b) )\in \mathcal{X}^{n}$ where $\mathbf{g} \in
  \{1,\cdots, k \}^{n}$. A single agent $k$ has preferences which
  are determined by bit $b$: if $b=0$, her favorite good is $g_k=k$
  her own, and if $b=1$, her favorite good is good $1$. In the case
  that $b=1$, agents $1, \ldots, k$ form a cycle with their favorite
  preferences: agent $i = 1, \cdots, k-1$ gets good $i+1$ and agent $k$ gets good 1, which would
  give each agent her favorite good. If, on the other hand, $b=0$,
  the uniquely IR trade is the $\pi(i) = g_i$, or no trade, since
  that instance contains no cycles, other than self loops.

  \begin{figure}[h]
    \begin{center}
    \begin{tabular}{|l|l|l|l|}
      \hline
      Agent & Endowment & Favorite Good & Second Favorite Good\\\hline
      1 & $1$ & $2$ & $1$ \\\hline
      2 & $2$ & $3$ & $2$ \\\hline
      \ldots & \ldots & \ldots & \ldots \\\hline
      $k-1$ & $k-1$ & $k$ & $k-1$ \\\hline
      $k$ (if $b = 0$) & $k$ & $k$ & Does not matter \\ \hline
      $k$ (if $b = 1$) & $k$ & $1$ & $k$ \\ \hline
      $k+1\ldots n$ & $k$ & $k$ & Does not matter\\ \hline
\end{tabular}
\end{center}
\caption{The endowments and preferences for Theorem~\ref{thm:lb-marginal}}
  \label{fig:preferences}
\end{figure}

Consider the mechanism $M_i(b) = M_S(\mathbf{x}(b))_i$ for $i \neq k$
(which is $(\epsilon, \delta)$-differentially private, since $M_S$ is
marginally differentially private). Let $f_1: \mathcal{G} \to \{0,1\}$ be the indicator function for
the event that $1$ receives good $2$; e.g. $f_1(1) = 0$ and
$f_1(2)=1$. Then, define $M'(b) = f_1(M_1(b))$, which is
also $(\epsilon, \delta)$-differentially private, due to $f$ being a post processing function. By individual
rationality, $M'(0) = 0$ with probability $1$.  Thus, by
Claim~\ref{claim:bit},

\[\Prob[M'(1) = 1] \leq\frac{e^\epsilon+\delta}{e^\epsilon + 1} \]

Thus,

\begin{align}
\Prob[M'(1) = 0] \geq 1 - \frac{e^\epsilon+\delta}{e^\epsilon + 1} \label{eq:lb}
\end{align}

When $M'(1) = 0$, $M_i(1) = g_i$ (each agent was
allocated her initial endowment). Consider the allocation $\pi(i) =
i+1 \mod k$ for $i\in [k]$ and $\pi(i) = g_i = k$ for $i= k+1, \cdots, n$. Agents $1, \ldots, k$ prefer $\pi$ to $M_S(\mathbf{x}(1))$
when $M'(1) = 0$, and all other agents are indifferent.  Since $M_S$
is $\alpha$-Pareto optimal, but there exists some $\pi$ which $k$ agents prefer
and no agent likes less,

\[\alpha\geq(1-\beta)\frac{k}{n}\Prob[M'(1) = 0] \geq
(1-\beta)\frac{k}{n}\left(1-\frac{e^\epsilon+\delta}{e^\epsilon +
    1}\right)\]

by Equation~\ref{eq:lb}, which completes the proof.
\end{proof}

\proof[\ifnum\final=0 Proof \fi of Corollary \ref{cor:copies}]
With the same hypotheses as Theorem \ref{thm:lb-marginal} we know that the number $k$ of types of goods must satisfy
 \begin{equation}
 k \leq \frac{n\alpha(e^\epsilon+1)}{(1-\delta)(1-\beta)}.
 \label{eq:copies}
 \end{equation}

 Hence, we must have some good with at least $\frac{(1-\delta)(1-\beta)}{\alpha(e^\epsilon+1)}$ copies, otherwise if we let $n_j$ be the number of goods of type $j \in [k]$ we have
 $$
 n = \sum_{j=1}^k n_j < k\left(\frac{(1-\delta)(1-\beta)}{\alpha(e^\epsilon+1)}\right) \underbrace{\leq}_{\text{\eqref{eq:copies}}} n
 $$ \endproof

\section{Appendix - Formal Algorithm PTTC}\label{sec:appendix}
\ifnum\final=0
\begin{algorithm}[!]
\caption{Private Top Trading Cycles}\label{PTTC}
\begin{algorithmic}[0]
\INPUT : Exchange Market $\mathbf{x} = (\mathbf{g},\succ)$.
\Procedure {\pttc} {$\bbx$}
\State Parameters: $\beta, \delta_1, \delta_2,\epsilon >0$
\State $\epsilon' =\frac{\epsilon \log(k^3/\beta)}{2\sqrt{8}\left( \log(k^3/\beta)\sqrt{k\log(1/\delta_1)} + k\sqrt{k\log(1/\delta_2)} \right) }$ and $E =  \frac{\log(k^3/\beta)}{\epsilon'}$.
\State Initialize $t\gets 1$, $\tau\gets 0$, and $P_e(1,0) \gets P_e$ is an ordered set of all $n$ agents by their index.
\begin{enumerate}
\item (Arc Weight Noise) For each $e \in A$ we set 
$$ \hat w_e(t,\tau) = w_e(t,\tau) + Z_e^t - 2E \qquad \text{ where } Z_e^t \sim Lap(1/\epsilon').$$
\item (Clear Cycle) while there is a cycle with positive weight, set  $\tau \gets \tau+1$

Denote the cycle by $C(t,\tau)$ and let $\hat W(t,\tau) \gets \min_{e \in C(t,\tau)} \{ \lfloor \hat w_e(t,\tau) \rfloor \}$. 
\item(Trade) For $e = (u,w) \in C^{\tau}_t$ 
\end{enumerate}
 \State $\qquad $ Set  $S_e(t,\tau)\gets \rselect(\hat W(t,\tau),P_e(t,\tau))$ and update:
  			$$
				P_e(t,\tau) \gets P_e(t,\tau-1) \backslash S_e(t,\tau)  \qquad
			$$
			$$
				w_e(t,\tau) \gets  w_e(t,\tau-1) - \hat W(t,\tau) \quad \& \quad \hat w_e(t,\tau) \gets \hat w_e(t,\tau-1) - \hat W(t,\tau)
			$$
			$$
				\pi(i) = v \quad \forall i \in S_e(t,\tau) \text{ where } g_i = u.
			$$
\begin{enumerate}
\setcounter{enumi}{3}
		\item(No Cycle - Deletion) If there is no cycle with positive rounded down noisy weight, then there exists a node $v$ s.t. $\hat n_v(t,\tau) = \sum_{e:e=(v,w)} \hat w_e (t,\tau) <k$ (Lemma \ref{lem:delete}).  Let
			$$
			IN_v = \{e \in A : e = (u,v) \text{ some } u \in V^{t} \} \quad \& \quad OUT_v = \{e \in A: e = (v,u) \text{ some } u \in V^{t} \}
			$$ 
\end{enumerate}
			$\qquad \qquad$ We then update:
			$$
				V^{t+1} \gets V^t\backslash\{v\} \qquad \& \qquad A \gets A \backslash \{OUT_v \cup IN_v\}.
			$$
			$\qquad \qquad$ For all $e = (u,w) \in A$  \\
				$\qquad \qquad \qquad$ Define $B_e(t) = \{ i \in P_{(u,v)}(t,\tau) : w \succeq_i g \quad \forall g \in V^{t+1}\}$, update:
			$$
				w_e(t+1,0) \gets w_e(t,\tau) + |B_e(t)| \quad \& \quad P_e(t+1,0) \gets P_e(t,\tau) \cup B_e(t)
			$$
				$\qquad \qquad\qquad $ and assign goods
			$$
				\pi(i) = g_i \qquad \forall i \in P_e(t,\tau) \text{ where } e \in OUT_v.  
			$$
			$\qquad \qquad$Set $t \gets t+1$ and $\tau \gets 0$.  Return to Step 1.  
\begin{enumerate}
\setcounter{enumi}{4}
		\item(Clean-up for IR)
\end{enumerate}
		$\qquad \qquad$ If $w_e(t,\tau)<0$ then set $\pi(i) = g_i$ $\forall i$ and HALT.\\
		$\qquad \qquad$ Otherwise, return to step 1.\\
\Return $\pi$
\EndProcedure
\end{algorithmic}
\end{algorithm}

\else

\begin{algorithm}[h]
\SetAlgoNoLine
\pttc ($\mathbf{g},\succ $) \;
\KwIn{Exchange Market $\mathbf{x} = (\mathbf{g},\succ)$.}
\KwOut{An allocation to all agents}
{
	Parameters: $\beta, \delta_1,\delta_2,\epsilon >0$ \\
	Set $\epsilon' =\frac{\epsilon \log(k^3/\beta)}{2\sqrt{8}\left( \log(k^3/\beta)\sqrt{k\log(1/\delta_1)} + k\sqrt{k\log(1/\delta_2)} \right) }$ and $E =  \frac{\log(k^3/\beta)}{\epsilon'}$\\
	Initialize $t\gets 1$, $\tau\gets 0$, and $P_e(1,0) \gets P_e$ is an ordered set of all $n$ agents by their index.
	\begin{enumerate}
		\item (Arc Weight Noise) For each $e \in A$ we set 
		$$ \hat w_e(t,\tau) = w_e(t,\tau) + Z_e^t - 2E \qquad \text{ where } Z_e^t \sim Lap(1/\epsilon').$$ 
		\item (Clear Cycle) {\bf while }  there is a cycle with positive weight {\bf do}  $\tau \gets \tau+1$
		\\
		Denote the cycle by $C(t,\tau)$ and let $\hat W(t,\tau) \gets \min_{e \in C(t,\tau)} \{ \lfloor \hat w_e(t,\tau) \rfloor \}$.
		\item (Trade) { \bf for } $e = (u,w) \in C(t,\tau)$ {\bf do} \newline
 			$\qquad $ Set  $S_e(t,\tau)\gets \rselect(\hat W(t,\tau),P_e(t,\tau))$.   
 			Update:
 			$$
				P_e(t,\tau) \gets P_e(t,\tau-1) \backslash S_e(t,\tau)  \qquad
			$$
			$$
				w_e(t,\tau) \gets  w_e(t,\tau-1) - \hat W(t,\tau) \quad \& \quad \hat w_e(t,\tau) \gets \hat w_e(t,\tau-1) - \hat W(t,\tau)
			$$
			$$
				\pi(i) = v \quad \forall i \in S_e(t,\tau) \text{ where } g_i = u.
			$$
		\item(No Cycle - Deletion) If there is no cycle with positive rounded down noisy weight, then there exists a node $v$ s.t. $\hat n_v(t,\tau) = \sum_{e:e=(v,w)} \hat w_e (t,\tau) <k$ (Lemma \ref{lem:delete}).  Let
			$$
			IN_v = \{e \in A : e = (u,v) \text{ some } u \in V^{t} \} \quad \& \quad OUT_v = \{e \in A: e = (v,u) \text{ some } u \in V^{t} \}
			$$ 
			We then update:
			$$
				V^{t+1} \gets V^t\backslash\{v\} \qquad \& \qquad A \gets A \backslash \{OUT_v \cup IN_v\}.
			$$
			{\bf for } all $e = (u,w) \in A$  {\bf do } \\
				$\qquad$ Define $B_e(t) = \{ i \in P_{(u,v)}(t,\tau) : w \succeq_i g \quad \forall g \in V^{t+1}\}$ and update:
			$$
				w_e(t+1,0) \gets w_e(t,\tau) + |B_e(t)| \quad \& \quad P_e(t+1,0) \gets P_e(t,\tau) \cup B_e(t)
			$$
				$\qquad $ and assign goods
			$$
				\pi(i) = g_i \qquad \forall i \in P_e(t,\tau) \text{ where } e \in OUT_v.  
			$$
			Set $t \gets t+1$ and $\tau \gets 0$.  Return to Step 1.  
		\item(Clean-up for IR) \\
		$\qquad \qquad$ {\bf if } $(w_e(t,\tau)<0)$ {\bf then } set $\pi(i) = g_i$ $\forall i$ and HALT.\\
		$\qquad \qquad$ {\bf else } Return to step (1).
		
	\end{enumerate}
\Return $\pi$
}
\caption{Private Top Trading Cycles}\label{PTTC}
\end{algorithm}
\fi

\section{Formal Proofs from Private Top Trading Cycles Section \ref{SEC:PTTC} }\label{sec:appendix_pttc}
\proof[\ifnum\final=0 Proof \fi of Claim \ref{lowerror}]
Recall that for a random variable $Z \sim $ Lap$(b)$ we have

\[\Prob(|Z| \geq \mu\cdot b) = e^{-\mu}\]

We then have, for $\mu = \log(k^3/\beta)$ and $b = \frac{1}{\epsilon'}$

\[\Prob\left( |Z| \geq \frac{ \log(k^3/\beta)}{\epsilon'}  \right) = \frac{\beta}{k^3}.\]

Now our algorithm in a fixed round $t$ will sample a new Laplace random variable $Z_e^t$ at
most $k^2$ times -  for each of the $O(k^2)$ arcs, one Laplace random
variable is sampled and there are as many as $k$ rounds according to Lemma \ref{lem:running_time}.
Hence, we can obtain the following bound,

\[\Prob\left( |Z_e^t| \leq \frac{ \log(k^3/\beta)}{\epsilon'} \quad
  \forall e \in A, \forall t \in [k] \right) \geq 1-\beta.\]

Hence we can then lower bound and upper bound the difference
between the error in the noisy arc weights and the actual arc weights
to get the relation in \eqref{eq:webound}:
\[E \leq w_e - \hat w_e= 2E - Z_e \leq 3E .\]
\endproof

\begin{proof}[\ifnum\final=0 Proof \fi of Lemma \ref{lem:mdp_lem1}]
We first fix any agent $i$ and agent types $\bbx=(\bbx_{-i},x_i)$ and $\bbx'=(\bbx_{-i},x'_i)$.  Let us first consider the noisy arc weights $\hat w(1,0)$ for round $t=1$.  We define $M^1:\cX^n \to \R^{k^2} $ as
$$
M^1(\bbx) = \hat w(1,0).
$$
We know that this computation just uses the Laplace Mechanism from Algorithm \ref{Lap_Mech} where we add Laplace noise with parameter $1/\epsilon'$ to each component of the exact edge weights $w(1,0)$.  When agent $i$ changes her data from $x_i$ to $x_i'$, she can change at most 2 entries in $w(1,0)$ by at most 1.  Thus, we have shown that $M^1$ is $2\epsilon'$-differentially private.  

Throughout round 1, agents will trade goods and be removed from consideration once they have a good.  To figure out how much the arc weights will be impacted in future rounds, we use the algorithm $\rselect$ that selects a given number of agents to trade (that is determined by $\hat w(1,0)$) from those that want to trade in a random way.  We denote here the random vector $\sigma_1^\tau$ where each of its entries correspond to the random value selected internally in $\rselect$.  This vector of values $\sigma_1^\tau$ determines who trades from each arc in cycle $C(1,\tau)$.    We write the randomness from round one as $r_1 = (\hat w(1,0),\sigma_1)$ where $\sigma_1 = (\sigma_1^\tau)_{\tau \in [k^2]}$.  

We will now assume that we have the randomness from all previous rounds, i.e. $\vec r_{t} = (r_1,\cdots, r_{t})$, and consider round $t+1$.  Our algorithm $\pttc$ can compute the exact arc weights at round $t+1$ as a function of these past random vectors and the data of the agents.  We now ask, what is the sensitivity of the exact arc weights at round $t+1$ when agent $i$ changes reports and we condition on all the prior randomization $\vec r_{t}$?  We consider the following cases when agent $i$ reports $x_i$ and compare to when $i$ reports $x_i'$ instead:
\begin{itemize}
\item[$\bullet$] Agent $i$ has yet to be selected after reporting $x_i$ at the start of round $t+1$.  
\begin{itemize}
\item[--] Agent $i$ has yet to be selected after reporting $x_i'$ at the start of round $t+1$.  In this case when we fix all the prior information, then the sensitivity in the actual arc count vector $w(t+1,0)$ is at most 2 because agent $i$ may move from one arc to another.
\item[--] Agent $i$ was selected at an earlier round $t'\leq t$ with report $x_i'$.  Fixing all the prior information, when agent $i$ was selected on cycle $C(t',\tau)$ for some $\tau$ when she reported $x_i'$, she must have caused some agent $i_{t',\tau}$ to not get chosen that was chosen when $i$ reported $x_i$.  This can further cascade to when agent $i_{t',\tau}$ gets chosen at a later cycle; someone that was selected when $i$ reported $x_i$ will not be selected when $i$ reports $x_i'$ instead.  However, once agent $i$ was selected when reporting $x_i'$, she has no further impact on the exact arc weights other than the person she displaced from her trade.  Further, when the people that were displaced by $i$ get matched, they are eliminated from further consideration.  Note that we are assuming that $i$ is still contributing weight to an arc when $i$ reported $x_i$, because $i$ has yet to be chosen when reporting $x_i$.  Thus in this case, conditioning on all the prior information, the sensitivity of the exact arc weight vector at round $t+1$ is at most $2$.  
\end{itemize}
\item[$\bullet$] Agent $i$ was selected at an earlier round, say $t' \leq t$ with report $x_i$.  
\begin{itemize}
\item[--] Agent $i$ was selected at an earlier round with report $x_i'$.  This follows a similar analysis as above when agent $i$ reports $x_i'$ and gets selected at round $t'\leq t$ when compared with $i$ getting chosen after round $t+1$ with report $x_i$.  
\item[--] Agent $i$ has yet to be selected at round $t+1$ with report $x_i'$.  We have already considered this case above, with the roles of $x_i$ and $x_i'$ switched.  
\end{itemize}
\end{itemize}
Thus, we have shown that conditioning on all the prior randomization, the exact arc counts $w(t+1,0)$ are $2$-sensitive.  

We now want to compute the noisy arc weights $\hat w(t+1,0)$ for round $t+1$ given $\vec r_{t}$.  We will denote $\cR^*$ to be the space where there can be an arbitrary number of random vectors, like $\vec r_{t}$.  We define $M^{t+1}:\cX^n \times \cR^* \to \R^{k^2}$ to be 
$$
M^{t+1}(\bbx, \vec r_{t}) = \hat w(t+1,0).  
$$

Hence for any given randomization vectors $\vec r_{t}$ the mechanism $M^{t+1}(\bbx,\vec r_t)$ is $2 \epsilon'$-differentially private in the database $\bbx$.  

The full mechanism $ \cM: \cX^n \to \R^{k\times k^2}$ that computes all the noisy arc weights in every round of $\pttc$ can be defined as 
$$
 \cM(\bbx) = (\hat w(t,0))_{t \in [k]}.
$$

We then apply advanced composition from Theorem \ref{thm:advanced_comp} and use the fact that the random integers $\sigma_t$ computed in $\rselect$ are chosen independent of the data to get that $ \cM$ is $(\epsilon_1,\delta_1)$-differentially private where $\delta_1>0$ and 
$$
\epsilon_1 = 2\epsilon' \cdot\sqrt{8k\log(1/\delta_1)}
$$
\end{proof}

\begin{proof}[\ifnum\final=0 Proof \fi of Lemma \ref{lem:mdp_lem2}]
We assume that we have the arc weights $\hat w$ for every round.  Consider round $t$ and cycle $\tau$ that has player $j \in P_e(t,\tau)$ on an arc $e \in C(t,\tau)$.  The probability that agent $j$ is selected depends on the exact number of people on edge $e$, that is $w_e(t,\tau)$ and the number of people $\hat W(t,\tau)$ that are being cleared along that cycle $C(t,\tau)$.  Recall in the proof of Lemma \ref{lem:mdp_lem1} we showed that the weights $w_e(t,0)$ at the beginning of round $t$ were a function of the reported types $\bbx$ and the prior randomization, that we called $\vec r_{t-1}$.  The randomization terms include the noisy arc weights for rounds before $t$, which are part of $\hat w$, but it also includes the random integers $\sigma$ calculated in $\rselect$ which is computed independent of the data.  

In order to find the exact number of people at an arc along cycle $C(t,\tau)$, we need to know the prior randomization terms and the random values $\sigma_t^1,\cdots, \sigma_t^{\tau-1}$. Further, we showed in the previous lemma that each entry $w_e(t,\tau)$ for $e \in C(t,\tau)$ may change by at most 1 when agent $i$ changes her reported type.  We let $w'_e(t,\tau) \in \{ w_e(t,\tau)-1, w_e(t,\tau), w_e(t,\tau)+1\}$ denote the number of people on arc $e \in C(t,\tau)$ when agent $i$ changes her type to $x_i'$, but the types of the other agents remains the same.  We will write the randomization terms as $\vec r_{t,\tau} = (\vec r_{t-1}, \hat w(t,0),\sigma_t^1,\cdots,\sigma_t^{\tau-1}) $.   We define the mechanism $M_j^{t,\tau}: \mathcal{X}^{n-1}\times\mathcal{X} \times R^* \to \{0\cup \G\}$ as
$$
M_j^{t,\tau}(\mathbf{x}_{-j},x_j, \vec r_{t,\tau}) = \left\{ \begin{array}{lr}
g & \text{ If } j \in S_e(t,\tau) \text{ for  $e= (g_j,g) \in C(t,\tau)$ } \\
0   & \text{ otherwise} 
							\end{array}
	\right.
$$
We will now make use of our assumption that the difference between the noisy arc weights and the actual arc weights satisfies \eqref{eq:webound}.  This gives us
\begin{align}
 \frac{\Prob\left( M_j^{t,\tau}(\mathbf{x}_{-j},x_j, \vec r_{t,\tau}) = g \right) }{\Prob\left(M_j^{t,\tau}((\mathbf{x}_{-i,j},x_i'),x_j, \vec r_{t,\tau})= g\right) } &  = \frac{\frac{\hat W(t,\tau)}{w_e(t,\tau)}}{\frac{\hat W(t,\tau)}{w_e'(t,\tau)}} \leq \frac{\frac{\hat W(t,\tau)}{w_e(t,\tau)}}{\frac{\hat W(t,\tau)}{w_e(t,\tau) + 1}} \nonumber \\
 & = \frac{w_e(t,\tau)+1}{w_e(t,\tau)} \leq 1+ \frac{1}{E} \leq e^{1/E}
 \label{eq:selected}
\end{align}
The first inequality follows from $w'_e(t,\tau) \leq w_e(t,\tau)+1$.  The second to last inequality comes from the
fact that the cycle always has an arc with $\hat w_e = w_e + Z_e - 2E \geq 1
\implies w_e \geq E+1$.  

We now consider the case when $j$ is not selected at round $t$, cycle $\tau$ when his arc is on the cycle.
\begin{align}
 \qquad\qquad &\frac{\Prob\left(M_j^{t,\tau}(\mathbf{x}_{-j},x_j, \vec r_{t,\tau}) =0 \right) }{\Prob\left( M_j^{t,\tau}((\mathbf{x}_{-i,j},x_i'),x_j, \vec r_{t,\tau}) =0\right) }   = \frac{1-\frac{\hat W(t,\tau)}{w_e(t,\tau)} }{1-\frac{\hat W(t,\tau)}{w_e'(t,\tau)}} \leq \frac{1-\frac{\hat W(t,\tau)}{w_e(t,\tau)}}{1-\frac{\hat W(t,\tau)}{w_e(t,\tau)-1}}  \nonumber \\
 & = \frac{(w_e(t,\tau) - \hat W(t,\tau)) (w_e(t,\tau) - 1)}{w_e(t,\tau) \cdot (w_e(t,\tau) - \hat W(t,\tau)-1)} \leq \frac{w_e(t,\tau) - \hat W(t,\tau)}{w_e(t,\tau) - \hat W(t,\tau) - 1}.
 \label{eq:intermediate}
\end{align}
The first inequality holds because $w'_e(t,\tau) \geq w_e(t,\tau) - 1$.  Recall that we have $w_e(t,\tau) -
\hat w_e(t, \tau) \geq E$ from \eqref{eq:webound} $\implies w_e(t,\tau) - \hat W(t,\tau) \geq E$.
Hence, we can further bound our ratio in \eqref{eq:intermediate} by
\begin{equation}
 \frac{\Prob\left(M_j^{t,\tau}(\mathbf{x}_{-j},x_j, \vec r_{t,\tau}) =0 \right) }{\Prob\left( M_j^{t,\tau}((\mathbf{x}_{-i,j},x_i'),x_j, \vec r_{t,\tau}) =0\right) } \leq \frac{E}{E-1} \leq 1+ \frac{2}{E} \leq e^{2/E}.
 \label{eq:notselected}
\end{equation}
Hence, fixing all randomness prior to selecting people at cycle $C(t,\tau)$, each $M_j^{t,\tau}$ is $2/E$-differentially private with respect to the database $\mathbf{x}_{-j}$ for all $\tau \in [k^2]$ and $t \in [k]$.  The mechanism $\cM_j$ in the statement of the lemma is then just a composition of the mechanisms $M_j^{t,\tau}$ over all rounds $t \in [k]$ and cycles $\tau \in [k^2]$ \footnote{In the theorem statement we are conditioning on the weights $\hat w$ instead of all prior randomness $\vec r$.  The only additional terms in $\vec r$ that are not $\hat w$ are the values $\sigma$ that determine who trades at each cycle in all prior rounds.  However, these values $\sigma$ are chosen independently of our data and so do not change our privacy analysis if we condition on $\vec r$ or $\hat w$.}.  We have shown that $M_j^{t,\tau}$ is $2/E$-differentially private in data $\bbx_{-j}$.  Thus, applying advanced composition in Theorem \ref{thm:advanced_comp}, we get that $\cM_j$ is $(\epsilon_2,\delta_2)$ differentially private for $\delta_2>0$ and 
$$
\epsilon_2 = 2k \sqrt{8k \log(1/\delta_2)}/E
$$
\end{proof}

\end{document}